\documentclass{article}

\newcommand*{\affaddr}[1]{#1} 
\newcommand*{\affmark}[1][*]{\textsuperscript{#1}}

\usepackage{amsmath}
\usepackage{amssymb}
\usepackage{amsfonts}
\usepackage{a4wide}
\usepackage{amsthm}
\usepackage{tikz}

\def\Kgg1{{\alpha_0}}
\def\1b{{\alpha_1}}
\def\Kgg1{{\beta_0}}
\def\2b{{\beta_1}}
\def\gg3{{\gamma_0}}
\def\3b{{\gamma_1}}
\def\gg4{{\epsilon_0}}
\def\4b{{\epsilon_1}}
\def\gg5{{\zeta_0}}
\def\5b{{\zeta_1}}
\def\gg6{{\mu_0}}
\def\6b{{\mu_1}}

\def\C{\mathbb C}

\def\mod{\text{ mod }}
\def\m{{\frak m}}

\def\x{{\frak x}}

\def\y{{\frak y}}

\def\p{{\frak p}}
\def\q{{\frak q}}

\def\rN{{\mathrm N}}

\newcommand{\beq}{\begin{equation}}
\newcommand{\eeq}{\end{equation}}
\newcommand{\la}{\lambda}
\newcommand{\al}{\alpha}
\newcommand{\be}{\beta}
\newcommand{\ga}{\gamma}
\newcommand{\ze}{\zeta}
\newcommand{\rA}{{\mathrm{A}}}
\newcommand{\rB}{{\mathrm{B}}}
\newcommand{\rC}{{\mathrm{C}}}
\newcommand{\Pro}{{\mathbb P}}

\newcommand{\Z}{{\mathbb Z}}

\newcounter{NN}
\newtheorem{notation}[NN]{Notation}

\newtheorem{remark}[NN]{Remark}
\newtheorem{proposition}[NN]{Proposition}
\newtheorem{theorem}[NN]{Theorem}
\newtheorem{corollary}[NN]{Corollary}

\newtheorem{lemma}[NN]{Lemma}
\newtheorem{definition}[NN]{Definition}

\def\m{{\frak m}}

\begin{document}
\bibliographystyle{plain}
\title{QRT maps and related Laurent systems}
\author{%
K. Hamad\affmark[1], A.N.W. Hone\affmark[2], P.H. van der Kamp\affmark[1], and G.R.W. Quispel\affmark[1] \\
\affaddr{\affmark[1]Department of Mathematics and Statistics}\\
\affaddr{La Trobe University, Victoria 3086, Australia}\\
\affaddr{\affmark[2]School of Mathematics, Statistics \& Actuarial Science}\\
\affaddr{University of Kent, Canterbury CT2 7NF, UK}%
}
\date{}

\maketitle

\begin{abstract}
In recent work 
it was shown how recursive factorisation of certain QRT maps leads to Somos-4 and Somos-5 recurrences with periodic
coefficients, and to a fifth-order recurrence with the Laurent property. Here we recursively factorise the 12-parameter symmetric QRT
map, given by a second-order recurrence, to obtain a
system of three coupled recurrences which possesses the Laurent property.
As degenerate special cases, we first 
derive
systems of two coupled recurrences corresponding to the 5-parameter multiplicative  and  additive symmetric QRT maps.
In all cases, the Laurent property
is established using a generalisation of a result due to Hickerson, and  exact formulae for degree growth  are found from  ultradiscrete (tropical) analogues of the recurrences. 
For the general 18-parameter QRT map it is shown that the components of the iterates 
can be written as a
ratio of quantities that satisfy 
the same Somos-7 recurrence. 
\end{abstract}

\section{Introduction}
A rational recurrence relation  is said to have the {\em Laurent property} if all of the iterates are Laurent polynomials
in the initial values, with coefficients belonging to some ring (typically $\Z$). We call such a  recurrence a {\em Laurent recurrence}. The first examples of such recurrences were discovered by
Michael Somos in the 1980s \cite{Gal}. Since then many more have been found \cite{ACH,FZ,fm,LP,Mas} (also cf. \cite{EZ}). The Laurent property is a central feature of cluster algebras (see \cite{fziv, fst} and references).

This paper is concerned with systems of Laurent recurrences related to QRT maps. The QRT maps 
 are an 
18-parameter family of birational transformations of the plane, which were 
introduced in \cite{QRT1,QRT2},  encompassing various examples that  appeared previously in a wide 
variety of contexts, including  
statistical mechanics, discrete soliton theory and dynamical systems. QRT maps are measure-preserving (symplectic) 
and have an invariant function (first integral),  hence they provide a prototype of  a discrete integrable system in finite dimensions. The generic level set of the  first integral is a curve of genus one, so there is an associated elliptic fibration of the plane \cite{tsuda}. The  rich geometry of QRT maps  is described extensively in the monograph by Duistermaat \cite{duistermaat}; for a terse overview, see subsection 6.3 below. 

It is an open question as to what conditions are necessary for ``Laurentification'' of a general birational transformation,  i.e. to determine whether such a transformation admits a lift to a Laurent  recurrence or a system of such recurrences.
In \cite{HK} two of the authors used ultradiscretization and recursive factorisation (which was
employed in \cite{vdK}, but can in fact be found in earlier work by  Boukraa and Maillard \cite{Mai})
to derive recurrence relations for the divisors of iterates of homogenised discrete  integrable systems. As the divisors
are polynomials, these recurrences should possess the Laurent property, as indeed they do,
in all cases considered. A different approach using projective coordinates has been taken in \cite{Via},
leading to similar results.

Specifically, it was shown  in \cite{HK}  that  two particular  multiplicative symmetric QRT maps, namely
\begin{subequations} \label{fh}
\begin{align} 
u_{n+1}u_{n-1}&=\frac{\alpha u_{n}+\beta}{u_n^{2}}, \label{fha} \\
u_{n+1}u_{n-1}&=\frac{\gamma u_n + \delta}{u_{n}} , \label{fhb}
\end{align}
\end{subequations}
give rise via recursive factorisation to Somos-4 and Somos-5 recurrences, that is
\begin{subequations} \label{S45}
\begin{align}
c_{n-2}c_{n+2}&=\alpha_{n}c_{n-1}c_{n+1}+\beta_{n}c_{n}^{2}, \label{S4} \\
d_{n-3}d_{n+2}&=\gamma_{n}d_{n-2}d_{n+1}+\delta_{n}d_{n-1}d_{n}, \label{S5}
\end{align}
\end{subequations}
respectively,
where the coefficients $\alpha_{n},\beta_{n}$ and $\gamma_{n},\delta_{n}$ are periodic functions of $n$,
with period 8 in the first case and period 7 in the second.  The  connection between the QRT maps
(\ref{fh}) and the
autonomous versions of these Somos recurrences is well known (see e.g. \cite{honeblms,honetams}).
Both equations (\ref{S45}) are special cases of a non-autonomous Gale-Robinson
recurrence \cite{Gal}, which arise as reductions of Hirota's bilinear (discrete KP) equation \cite{GR, Mas, mase, WTS, Zab}. Furthermore, it was shown in \cite{HK} that the additive QRT map
\begin{equation} \label{um}
u_{n+1}+u_{n-1} = \frac{\alpha-u_{n}^2}{u_{n}},
\end{equation}
known as
DTKQ-2 \cite{DTKQ},
is related by  recursive factorisation to a fifth-order Laurent recurrence, that is
\begin{equation} \label{fol}
e_{n-1}e_{n-2}^{2}e_{n-5}+e_{n-1}^{2}e_{n-4}^{2}+e_{n}e_{n-3}^{2}e_{n-4}=\alpha e_{n-2}^{2}e_{n-3}^{2}.
\end{equation}

It is worth pointing out that the Laurent property is neither necessary nor sufficient for integrability. To see why it is not necessary, note that a discrete integrable system, in the form of a birational map satisfying the conditions of Liouville's theorem, need not have the Laurent property: this property is associated with a particular choice of coordinate system, and is easily destroyed by a birational change of coordinates, whereas integrability is not. As for sufficiency, it is known that 
large families of birational recurrences with the Laurent property arise from certain sequences of mutations in a 
cluster algebra \cite{FZ, fm} or an LP algebra \cite{ACH, LP},  yet  integrability is a rare property, and only a small 
minority of such recurrences are discrete integrable systems. 

Nevertheless, in an algebraic setting, based on the evidence of a large number of examples, it appears  that 
discrete integrable systems should always admit Laurentification. The advantage of having a system with the  Laurent 
property is that it leads to a very direct way of calculating the sequence of degrees, so that the 
algebraic entropy of the system can be calculated as the limit $\lim_{n\to\infty} n^{-1} \log d_n$ (where 
$d_n$ is the degree of the $n$th iterate). In the approach of Bellon and Viallet \cite{bv}, 
discrete integrable systems are characterised by having zero algebraic entropy. For the case of the QRT maps considered 
here, which can be regularised by a finite number of blowups of the plane \cite{duistermaat}, and preserve a pencil 
of invariant curves, general arguments indicate that the degrees grow quadratically with $n$ \cite{bellon}, and thus the
entropy is zero. While a geometrical approach via blowups is effective for counting degrees in dimension two, it becomes 
increasingly difficult in higher dimensions, and this is our motivation for considering Laurent systems here, in a test case 
where we know the degree growth in advance.

Laurentification is not a unique procedure, and for convenience one should aim to find the simplest system which has the Laurent property. Recursive factorisation can provide a Laurent system, but not always  the simplest one: 
in particular, as shown below (see also \cite{mase}), solutions to the non-autonomous Somos recurrences (\ref{S45}) from \cite{HK} 
are related to those of their autonomous versions, which are simpler. 
In section 3 we obtain a two-component autonomous system 
directly from the 5-parameter multiplicative symmetric QRT map
\begin{align}\label{mul}
u_{n+1}u_{n-1}=\frac{ a_{3} u_n^2+ a_{5} u_n + a_{6} }{a_{1} u_n^2+a_{2} u_n + a_{3}},
\end{align} 
by writing the iterates as a ratio $u_n=k_n/l_n$. 
We prove that this is a Laurent 
system, and use 
the Laurent property together with ultradiscretization to derive a polynomial formula for  the growth of degrees (quadratic in $n$). 
We also show how our autonomous system degenerates to the non-autonomous Somos-4 (\ref{S4}) and Somos-5 (\ref{S5}) in the special cases considered in \cite{HK}.

In section 4, we Laurentify the 5-parameter additive symmetric QRT map
\begin{align}\label{add}
u_{n+1}+u_{n-1}=-\frac{ a_{2} u_n^2+ a_{4} u_n + a_{5} }{a_{1} u_n^2+a_{2} u_n + a_{3}},
\end{align}
which generalises (\ref{um}). This gives another system of two recurrences, which degenerates to  (\ref{fol}) as a special case.
We show that the same quadratic formula as found for (\ref{mul}) describes the degree growth of (\ref{add}).

In section 5, we recursively factorise the 12-parameter symmetric QRT map 
(see equation (\ref{QRT}) below),
and obtain a three-component system, whose Laurentness follows directly from factorisation properties.
We describe how the additive and
multiplicative Laurent systems obtained from (\ref{mul}) and (\ref{add}) appear as  degenerate cases, and 
use ultradiscretization to show that the degree growth of the symmetric QRT map is quadratic. 

In section 6, we present Somos-7 recurrences that are satisfied by the variables in the  Laurent systems introduced in the preceding sections. We prove that the components of iterates of the general 18-parameter QRT map can also be written
as a ratio of quantities that satisfy a Somos-7 relation.

Because deriving systems that are likely to possess the Laurent property can
now be done routinely, there is a need for verifying the Laurent property routinely.
An account of such a procedure for autonomous recurrences, found by Hickerson, was given in \cite{Gal, Rob}.
Another approach, built into the  axiomatic framework of cluster algebras or LP algebras \cite{LP}, is to use the Caterpillar Lemma as in \cite{FZ},
but this only applies to relations in multiplicative form (i.e. exchange relations with a product of two terms on the left-hand side).
A straightforward generalisation of Hickerson's method to systems of equations, with more general denominators,
is given  in Theorem \ref{T1} in the next section.  
For the multiplicative
and additive Laurent systems (equations (\ref{smsql}) and (\ref{M12}) below) it is easy to verify the conditions in the theorem, and hence to establish their Laurentness.

\section{Proving the Laurent property}
 Sufficient conditions  for equations of the form
\begin{align}\label{rob}
  \tau_{n}\tau_{n-k}=P(\tau_{n-k+1},\dots,\tau_{n-1}), \qquad k \in \mathbb{N}, \quad P\,\, \mathrm{ polynomial} 
\, \mathrm{over}\, {\cal R} 
\end{align}  
(where $\cal R$ is a ring of coefficients) 
to possess the Laurent property were found by Hickerson.
Taking $\{\tau_{i}\}_{i=0}^{k-1}$ as the initial values, the iterates are written as a ratio
$$\tau_{n}=\frac{p_{n}(\tau_{0},\dots,\tau_{k-1})}{q_{n}(\tau_{0},\dots,\tau_{k-1})}$$
of coprime polynomials, so that the greatest common divisor
$(p_{n},q_{n})=1$.  The Laurent property means that all $q_n$ are monomials.
The following is Hickerson's result, as mentioned by Gale in \cite{Gal}
and proved by Robinson in \cite{Rob}.
\begin{theorem}
Equation (\ref{rob}) has the Laurent property if $(p_{k},p_{k+l})=1$ for $l =1,\dots,k$ and $q_{2k}$ is a monomial.
\end{theorem}

Below we provide sufficient conditions  for systems of equations to possess the Laurent property. At the same time, we generalise
the form of the right-hand side of (\ref{rob}), by allowing a monomial denominator, and consider the case 
where the iterates  are 
Laurent polynomials in a subset of the initial variables and polynomial in the rest. 

 Consider a  system of $d$ ordinary difference equations of order $k$,
\begin{align}\label{IM1}
\tau_{n}^{i}=\frac{P^{i}(\tau_{n-k}^{1},\dots,\tau_{n-1}^{d})}{Q^{i}(\tau_{n-k}^{1},\dots,\tau_{n-1}^{d})},\quad P^i \,\, \mathrm{ polynomial},
 \quad Q^i \,\,\mathrm{ monomial},  \qquad i=1,\ldots,d. 
\end{align}
From a set of $kd$ initial values $U=\{\tau_{l}^{j}\}_{1\leq j\leq d,0\leq  l \leq k-1}$, where the superscripts denote components (not exponents), one finds $\tau_{n}^{i}$ 
as  rational functions of the initial values, given by
\begin{align}\label{IM2}
\tau_{n}^{i}=\frac{p_{n}^{i}(\tau_{0}^{1},\dots,\tau_{k-1}^{d})}{q_{n}^{i}(\tau_{0}^{1},\dots,\tau_{k-1}^{d})},
\end{align} with $(p_{n}^{i},q_{n}^{i})=1$.  By definition, if $ q_{n}^{i}\in{\cal R}[U]$ is a monomial for all $i$ and $n\geq0$, then (\ref{IM1}) has the Laurent property, meaning that each $\tau_{n}^{i}$ belongs to the 
ring ${\cal R}[U^{\pm 1}]:={\cal R}[(\tau_{0}^{1})^{\pm 1},\dots,(\tau_{k-1}^{d})^{\pm 1}]$. The form of (\ref{IM1}) guarantees that all components $q_{n}^{i}$ are monomials for $0\leq n \leq k$. Suppose these monomials depend on a subset of the initial values $V\subset U$, specified by a set of superscripts $I\subset \{1,\ldots,d\}$. The following conditions guarantee that $q_{n}^{i}$ are monomials 
 for all $i$ and $n\geq0$.

\begin{theorem}\label{T1} Suppose that  $ q_{k}^{i}$ is a monomial in ${\cal R}[V]$ for $1\leq i\leq d$.
If $p_{k}^{i}$ is  coprime to $p_{k+l}^{j}$ for all $ i,j\in I\subset\{1,\ldots, d\}$, $l =1,\dots,k$, and $q_{m}^{i}\in {\cal R}[V]$ is a monomial for $1\leq i\leq d$,
$k+1\leq m \leq 2k$, then (\ref{IM1}) has the Laurent property: all iterates are Laurent polynomials in the variables from $V$ and they are polynomial in the remaining variables from $W=U\setminus V$.
\end{theorem}
\begin{proof}
The proof is by induction in $n$.
If we regard $\{\tau_{l}^{j}\}_{1\leq j\leq d,1\leq  l \leq k}$ as initial data, then from (\ref{IM2}) we may write
\begin{align}\label{IM3}
\tau_{n}^{i}=\frac{p_{n-1}^{i}(\tau_{1}^{1},\dots,\tau_{k}^{d})}{q_{n-1}^{i}(\tau_{1}^{1},\dots,\tau_{k}^{d})},
\end{align}
while on the other hand,
by  taking $\{\tau_{l}^{j}\}_{1\leq j\leq d,k+1\leq  l \leq 2k}$ as initial values, we find
\begin{align}\label{IM4}
\tau_{n}^{i}=\frac{p_{n-k-1}^{i}(\tau_{k+1}^{1},\dots,\tau_{2k}^{d})}{q_{n-k-1}^{i}(\tau_{k+1}^{1},\dots,\tau_{2k}^{d})}.
\end{align}
Then by using (\ref{IM2}) again, the arguments $\tau^j_{m}$ for $m=k,\ldots,2k$ can be expressed as
Laurent polynomials in the variables from $V$ with coefficients in ${\cal R}[W]$. 
Thus 
for each $i$ the denominator
of (\ref{IM3}) becomes a monomial in the variables from $V$, multiplied by powers of the polynomials $p^j_k$ for $j\in I$. On the other hand, the denominator of (\ref{IM4}) becomes a monomial in variables from $V$ only, multiplied by powers of  $p^j_{k+l}$ for $j\in I$ and $l=1,\ldots, k$.
By the coprimality assumption, the only way that these two expressions can be equal is if 
all the powers of polynomials $p^j_m$ for $j\in I$ appearing in a denominator cancel with the numerator in each case, to leave a reduced expression for $\tau_{n}^{i}$ as a Laurent polynomial in the ring ${\cal R}[W][V^{\pm 1}]$. 
\end{proof}

The preceding result can be modified to include the case where the coefficients in system (\ref{IM1}) are periodic functions, e.g. as in (\ref{S45}), but we will not need this in the sequel. However, when discussing ultradiscretization it will
be convenient to describe periodic sequences using the following notation.
\begin{notation} \label{not}
A periodic function $f_n$ such that $f_{n+m}=f_n$ is defined by $m$ values: we write $f_{\mod m}=[v_1,\ldots,v_m]$ to mean $f_n=v_{n \mod m}$.
\end{notation}

\section{The multiplicative symmetric QRT map}
In this section, we show how to ``Laurentify'' the multiplicative symmetric QRT map, i.e. produce a corresponding
Laurent system  of recurrences, by applying homogenisation.
We then use ultradiscretization  to derive the
degree growth of
the 
map.
We also show how the Laurent system 
reduces to the
Somos-4 and Somos-5 equations with  periodic coefficients that were found in \cite{HK}.

\subsection{Laurentification of the multiplicative symmetric QRT map}

By taking  $u_{n}=\frac{k_{n}}{l_{n}}$  in $(\ref{mul})$  and  identifying  the numerators and denominators on both sides, we obtain a system that generates
sequences
$(k_{n})$ and $(l_{n})$, that is
\begin{subequations} \label{smsql}
\begin{align}
{k}_{n+1}{k}_{n-1}&=a_{3}{k}_{n}^{2}+a_{5}  {k}_{n}{l}_{n} +a_{6} {l}_{n}^{2},\label{smsqla}\\
{l}_{n+1}{l}_{n-1}&=a_{1}  {k}_{n}^{2}+a_{2}  {k}_{n}{l}_{n} +a_{3} {l}_{n}^{2}\label{smsqlb}.
\end{align}
\end{subequations}
Without loss of generality one can choose $(k_0,k_1,l_0,l_1)=(u_0,u_1,1,1)$ as  initial values for (\ref{smsql}).
Observe that the system (\ref{smsql}) is homogeneous of degree 2: it is 
a Hirota bilinear form for 
 $(\ref{mul})$.

\begin{proposition} \label{thm5}
The system (\ref{smsql}) has the Laurent property. Any four adjacent iterates $k_n$, $l_n$, $k_{n+1}$, $l_{n+1}$ are pairwise coprime  Laurent polynomials in the ring
$
{\cal R}[k_0^{\pm 1},k_1^{\pm 1},l_0^{\pm 1},l_1^{\pm 1}],
$
where ${\cal R}=\Z [a_1,a_2,a_3,a_5,a_6]$ is the ring of coefficients.
\end{proposition}
\begin{proof}
The Laurent property can be verified directly by applying Theorem \ref{T1} in the case that the dimension $d=2$ and the order $k=2$.
For the coprimality, observe that when $n=0$ this is trivially true, and proceed by induction in $n$.
If a non-constant Laurent polynomial $P\in{\cal R}[k_0^{\pm 1},k_1^{\pm 1},l_0^{\pm 1},l_1^{\pm 1}]$ is a common factor of $k_{n+1}$ and $k_n$, then it
divides the right-hand side of (\ref{smsqla}), hence $P| l_n$, which contradicts $(k_n,l_n)=1$. Thus
$(k_{n+1},k_n)=(k_{n+1},l_n)=1$, and similarly from   (\ref{smsqlb}) we have $(l_{n+1},k_n)=(l_{n+1},l_n)=1$. Now let  $P$ be  a  common factor of $k_{n+1}$ and $l_{n+1}$, 
$$
\implies
\mathrm{S} \, {\bf h}_n = P \, {\bf v}_n,
\ \mathrm{ where } \ \
\mathrm{S}=\begin{pmatrix}
a_3 & a_5 & a_6 & 0 \\
0 & a_3 & a_5 & a_6 \\
a_1 & a_2 & a_3 & 0 \\
0 & a_1 & a_2 & a_3
\end{pmatrix}, \ {\bf h}_n =\left(\begin{array}{c} k_n^3 \\ k_n^2l_n \\ k_nl_n^2 \\ l_n^3 \end{array}\right),
\, P \, {\bf v}_n =\left(\begin{array}{c} k_nk_{n+1}k_{n-1} \\ k_{n+1}k_{n-1}l_n \\ k_nl_{n+1}l_{n-1} \\ l_{n+1}l_{n-1}l_n \end{array}\right)  .
$$ Multiplying the above equation by the adjugate of 
the Sylvester matrix $\mathrm S$ yields
$$
\mathrm{R} \,   {\bf h}_n = P\,  \mathrm{S^{adj} }  {\bf v}_n,
$$
where the resultant $\mathrm R=\det\mathrm{S}$ is a non-zero element of the coefficient ring $\cal R$, namely
\beq\label{resul}
\mathrm{R} = a_1^2 a_6^2-a_1 a_2 a_5 a_6-2 a_1 a_3^2 a_6+a_1 a_3 a_5^2+a_2^2 a_3 a_6-a_2 a_3^2 a_5+a_3^4\neq 0.
\eeq
Hence $P$ divides each component of the vector ${\bf h}_n$, contradicting $(k_n,l_n)=1$.
\end{proof}
\begin{remark}
The latter  result remains true for numerical values $a_i$ such that $a_1a_3a_6\mathrm{R}\neq 0$.
\end{remark}

The Laurent property implies  that, in general,  the iterates of (\ref{smsql}) can be written in the form
\beq\label{deform}
k_n = \frac{\rN_n ({\bf k})}{{\bf k}^{{\bf d}_n}}, \qquad
l_n = \frac{\hat{\rN}_n ({\bf k})}{{\bf k}^{{\bf e}_n}},
\eeq
where $\rN_n$, $\hat{\rN}_n$ are polynomials in ${\bf k}=(k_0,k_1,l_0,l_1)$ that are not divisible by any of these four variables,
while the denominators are Laurent monomials, i.e.
$$
{\bf k}^{{\bf d}_n}=k_0^{d_n^{(1)}}k_1^{d_n^{(2)}}l_0^{d_n^{(3)}}l_1^{d_n^{(4)}}, \qquad
{\bf d}_n = (d_n^{(1)}, d_n^{(2)}, d_n^{(3)}, d_n^{(4)})^T
$$
and similarly for ${\bf k}^{{\bf e}_n}$, where the exponents appearing in the denominator vectors
${\bf d}_n$
and ${\bf e}_n$ are integers. The initial vectors are
\beq\label{dvecs}
{\bf d}_0=\left( \begin{array}{c} -1 \\ 0 \\ 0 \\ 0\end{array}\right), \quad
{\bf d}_1=\left( \begin{array}{c} 0\\  -1  \\ 0 \\ 0\end{array}\right), \quad
{\bf e}_0=\left( \begin{array}{c}  0 \\ 0\\-1 \\ 0\end{array}\right), \quad
 {\bf e}_1=\left( \begin{array}{c}  0 \\ 0\\ 0 \\ -1 \end{array}\right).
\eeq

\subsection{Growth of degrees of the multiplicative symmetric QRT map}

In order to measure the growth of degrees of the map (\ref{mul}), we consider the growth of the degrees of the Laurent polynomials
that are generated by the system  (\ref{smsql}). From the form of this system, it is clear that $k_n,l_n$ are homogeneous rational functions of degree 1 in
the initial values   ${\bf k}=(k_0,k_1,l_0,l_1)$, which implies that
\beq\label{homog}
\deg_{\bf k}( \rN_n  ({\bf k})) = 1+\deg_{\bf k} ({\bf k}^{{\bf d}_n}) =1+\sum_{j=1}^4d_n^{(j)}, \qquad
\deg_{\bf k}( \hat{\rN}_n ({\bf k})) = 1+\deg_{\bf k} ({\bf k}^{{\bf e}_n}) =1+\sum_{j=1}^4e_n^{(j)},
\eeq
where $\deg_{\bf k}$ denotes the total degree  in these variables.
Furthermore, from the form of (\ref{smsql}), $k_n,l_n$ are also subtraction-free rational expressions in the initial values,
meaning that a standard argument which is used in the theory of cluster algebras can be applied (cf. equation (7.7) in \cite{fziv}, or Lemma 8.3 in \cite{fst}),
and hence the denominator vectors ${\bf d}_n,{\bf e}_n$ satisfy a tropical version of the Laurent system,
given by the max-plus ultradiscretization\footnote{Here, and in the sequel, when going from a discrete equation  (recurrence) to an ultradiscrete one, we assume the parameters $a_i$ are
generic; in particular, for (\ref{smsql}) we assume non-zero $a_i $ such that (\ref{resul}) holds.}
\beq\label{tropical}
\begin{array}{rcl}
{\bf d}_{n+1}+{\bf d}_{n-1} & = & \max (2{\bf d}_{n}, {\bf d}_{n}+{\bf e}_{n}, 2{\bf e}_{n}), \\
{\bf e}_{n+1}+{\bf e}_{n-1} & = & \max (2{\bf d}_{n}, {\bf d}_{n}+{\bf e}_{n}, 2{\bf e}_{n}),
\end{array}
\eeq
where the max applies componentwise on the right-hand side.

Due to  the symmetrical form of the tropical system (\ref{tropical}) and the initial vectors (\ref{dvecs}), the solution of this ultradiscrete vector
system can be written as
$$
{\bf d_n} = (d_n,d_{n-1},e_n,e_{n-1})^T, \qquad
{\bf e_n} =  (e_n,e_{n-1},d_n,d_{n-1})^T,
$$
in terms of a pair of sequences $(d_n)$, $(e_n)$ which satisfy the scalar version of (\ref{tropical}), that is
\beq\label{scalar}
\begin{array}{rcl}
{d}_{n+1}+{d}_{n-1} & = & \max (2{d}_{n}, { d}_{n}+{e}_{n}, 2{ e}_{n}), \\
{ e}_{n+1}+{ e}_{n-1} & = & \max (2{d}_{n}, {d}_{n}+{e}_{n}, 2{ e}_{n}),
\end{array}
\eeq
with the initial values
$ d_0=-1$,  $d_1=e_0=e_1=0$.
If we introduce the sums and differences
$$
\Sigma_n = d_n+e_n, \qquad \Delta_n = d_n-e_n,
$$
and note the fact that $\max(\Delta, 0, -\Delta)=|\Delta|$, then the scalar system (\ref{scalar}) becomes
\beq\label{simp}
\begin{array}{rcl}
\Sigma_{n+2}-2\Sigma_{n+1}+\Sigma_{n} & = & 2|\Delta_{n+1}|, \\
\Delta_{n+2}+\Delta_{n} & = & 0 ,
\end{array}
\eeq
with initial values
\beq\label{dsinits}
\Sigma_0=\Delta_0=-1, \qquad \Sigma_1=\Delta_1=0.
\eeq
The decoupled equation for $\Delta_n$ implies that this quantity has period 4, and from the initial values it is clear
that
\beq\label{delsol}
\Delta_{\mod 4}=[0,1,0,-1],
\eeq
so the right-hand side
of the first equation in (\ref{simp}) has period 2, which gives the homogeneous linear
equation
$$
({\cal S}^2 -1)({\cal S}-1)^2 \Sigma_n=0,
$$ where $\cal S$ denotes the shift operator such that ${\cal S}\Sigma_n = \Sigma_{n+1}$.
Using the fact that $\Sigma_n$ takes  the sequence of values $-1,0,1,4$ for $n=0,1,2,3$, this fourth-order
recurrence is readily solved.
\begin{lemma}\label{ivp}
The solution of the system (\ref{simp}) with initial values (\ref{dsinits}) is given by
$$
\Sigma_n = \frac{1}{2}n^2 -\frac{3}{4}-\frac{(-1)^n}{4}
$$
together with (\ref{delsol}).
\end{lemma}

Now if we substitute in the initial values $(k_0,k_1,l_0,l_1)=(u_0,u_1,1,1)$ then the numerators in (\ref{homog}) become a pair of polynomials in $u_0,u_1$, denoted $\rN_n ({\bf u})$, $\hat{\rN}_n({\bf u})$, and 
we find
\beq\label{unform}
u_n = \frac{k_n}{l_n}=\frac{\rN_n({\bf u})}{u_0^{\Delta_n}u_1^{\Delta_{n-1}} \hat{\rN}_n({\bf u})}.
\eeq
For generic non-zero coefficients such that (\ref{resul}) holds,   the polynomials $\rN_n ({\bf u})$
and $\hat{\rN}_n({\bf u})$ are coprime, and from the form of the recurrence they always contain a term of highest possible total degree in $u_0,u_1$. Thus, by  (\ref{homog}) and the result of Lemma \ref{ivp}, we have
$$
\deg_{\bf u}( \rN_n ({\bf u})) = \deg_{\bf k}   (\rN_n ({\bf k})) =\Sigma_{n-1}+\Sigma_n+1 =n^2-n,
$$
and the same formula holds for  $ \deg_{\bf u}( \hat{\rN}_n ({\bf u}) )$. From the periodic sequence  (\ref{delsol}) it is clear that the Laurent monomial factor  in (\ref{unform}) cycles in the pattern $u_0,u_1,u_0^{-1},u_1^{-1}$, so the degree of the numerator is one more than the degree of the denominator, or vice versa. Hence
$\deg_{\bf u}(u_n)=1+\deg_{\bf u}( \rN_n ({\bf u}))=1+\deg_{\bf u}( \hat{\rN}_n ({\bf u}))$, which yields an exact formula for this degree.
\begin{theorem}
As a rational function of the initial values $u_0,u_1$, the $n$th iterate $u_n$ of
 the multiplicative QRT map (\ref{mul}) has degree
$n^2-n+1$.
\end{theorem}

\subsection{Degeneration to Somos recurrences}	
In this subsection, we demonstrate that the Somos-4 and Somos-5 recurrences  (\ref{S45}),  with periodic coefficients of periods 8 and 7 respectively, arise as special cases of the Laurent system (\ref{smsql}).

In \cite{HK} the initial data for (\ref{fha}) were taken as $u_1,u_2$, while here we use $u_0,u_1$ instead; mutatis mutandis, the periodic coefficients found in \cite{HK} are defined as follows.
 \begin{definition}
The Somos-4 equation with periodic coefficients which generates the divisors of the QRT map (\ref{fha}) is given by
(\ref{S4})  with
$
\alpha_{n}=\alpha u_{0}^{\p_{n+2}} u_{1}^{\p_{n-1}}$, $\beta_{n}=\beta u_{0}^{\q_{n+2}} u_{1}^{\q_{n-1}} $,
and
\begin{equation} \label{pq}
\p_{\mod 8}=[0,1,0,0,1,0,0,1]\text{ and }\q_{\mod 8}=[2,0,0,1,0,1,0,0].
\end{equation}
\end{definition}
\begin{theorem} \label{Thm1}
In the degenerate case $a_{1}=1$, $a_{2}=a_{3}=0$, $a_{5}=\alpha$ and $a_{6}=\beta$ the quantities $k_{n}$ and $l_{n}$ can be expressed in terms of solutions to (\ref{S4}) as
\begin{align}\label{m15}
k_{n}=u_{0}^{2\ze_{n+2}+\q_{n+1}-\p_{n+1}}u_{1}^{2\ze_{n-1}+\q_{n-2}-\p_{n-2}}c_{n-2}c_{n}, \quad
l_{n}=u_{0}^{2\ze_{n+2}}u_{1}^{2\ze_{n-1}}c_{n-1}^{2},
\end{align}
where $\ze_n$ satisfies
\beq\label{zeta}
({\cal S}-1)^2\, \ze_n=\p_n-\q_n.
\eeq
\end{theorem}
\begin{proof}
Taking $a_{1}=1$ and  $a_{2}=a_{3}=0$, $a_{5}=\alpha$ and $a_{6}=\beta$ in equation (\ref{smsql}),
it is easy to see that $k_n=\tau_{n-2}\tau_n$, $l_n=\tau_{n-1}^2$ is a solution of (\ref{smsql}) whenever
$\tau_n$ satisfies the autonomous Somos-4 recurrence
\beq \label{auts4}
\tau_{n+2}\tau_{n-2}=\al \tau_{n+1}\tau_{n-1}+\be \tau_n^2.
\eeq
Moreover, for this degenerate choice of coefficients, every solution of (\ref{smsql}) can be written in this way, e.g. by taking $\tau_{-2}=k_0/\sqrt{l_1}$,  $\tau_{-1}=\sqrt{l_0}$, $\tau_0=\sqrt{l_1}$,  $\tau_{1}=k_1/\sqrt{l_0}$ as
initial values for (\ref{auts4}). Now one can make a gauge transformation between $\tau_n$ and $c_n$, which is
required to satisfy    (\ref{S4}) :
$$
\tau_n = u_0^{\ze_{n+3}}u_1^{\ze_{n}}c_n \implies ({\cal S}^2 -1)^2 \ze_n =-\q_{n+1},
$$
a fourth-order recurrence for $\ze_n$, and also (\ref{zeta}) must hold. But the form of the
 8-periodic coefficients (\ref{pq}) implies that $({\cal S}+1)^2 (\p_n-\q_n)=-\q_{n+1}$, so the fourth-order relation is
a consequence of (\ref{zeta}), which completely determines the sequence of exponents $(\ze_n)$, 
up to a suitable choice of
$\ze_0,\ze_1$.
\end{proof}

\begin{definition}
The Somos-5 equation with periodic coefficients which generates the divisors of the QRT map (\ref{fhb})
is given by (\ref{S5}),
where
\begin{align*}
\gamma_{n}=\gamma u_{0}^{\p_{n+1}}u_{1}^{\p_{n-3}} \text { and }
\quad  \delta_{n}=\delta u_{0}^{\q_{n+1}}u_{1}^{\q_{n-3}},
\end{align*} with $\p_{\mod 7}=[1,0,0,0,1,0,0] \text{ and } \q_{\mod 7} =[0,0,1,0,0,1,1]$.
\end{definition}
\begin{theorem}
In the degenerate  case  $a_{1}=a_{3}=0$, $a_{2}=1$, $a_{5}=\gamma$ and $a_{6}=\delta$, $k_{n}$ and $l_{n}$ are expressed in terms of solutions to (\ref{S5}) as
\begin{align}\label{Q6}
k_{n}=u_{0}^{\eta_{n+3} +\eta_{n}}u_{1}^{\eta_{n-1} +\eta_{n-4}}d_nd_{n-3},
 l_{n}=u_{0}^{\eta_{n+2} +\eta_{n+1}}u_{1}^{\eta_{n-2} +\eta_{n-3}}d_{n-1}d_{n-2},
\end{align}where
\beq\label{third} ({\cal S}^3 - {\cal S}^2-{\cal S}+1)\eta_n = \p_n-\q_n. \eeq
\end{theorem}
\begin{proof} This is similar to the proof of Theorem \ref{Thm1}. The autonomous Somos-5 relation
$$\tau_{n-3}\tau_{n+2}= \ga \tau_{n-2}\tau_{n+1}+\delta \tau_{n-1}\tau_n$$ solves this degenerate
case of (\ref{smsql}) by taking $k_n= \tau_{n-3}\tau_n$, $l_n = \tau_{n-2}\tau_{n-1}$, and then
a gauge transformation $\tau_n =u_0^{\eta_{n+3}}u_1^{\eta_{n-1}}d_n$ gives a corresponding solution of (\ref{S5}),
provided $\eta_n$ satisfies the  third-order linear relation (\ref{third}) above.
\end{proof}

\section{The additive symmetric QRT map}
The aim of this section is to Laurentify the additive QRT map (also called the McMillan map), and to establish a formula for the growth of degrees.

\subsection{Laurentification of the additive symmetric QRT map}
Substituting $u_{n}=\frac{k_{n}}{l_{n}}$ into (\ref{add}),
and identifying quadratic numerators and denominators on each side leads to the following associated Laurent system.
\begin{proposition} \label{addlau}
The system
\begin{align}\label{M12}
\begin{array}{rl}
k_{n+1}l_{n-1} 
&=-(
 k_{{n-1}}l_{{n+1}}+
a_{2}\,{k_{{n}}}^{2}+a_{4}\,k_{{n}}l_{{n}}+a_{5}\,{l_{{n}}}^{2}),
\\
\phantom{xxxlxxxx} l_{n+1}l_{n-1}&=a_{1}\,{k_{{n}}}^{2}+a_{2}\,k_{{n}}l_{{n}}+a_{3}\,{l_{{n}}}^{2}
\end{array}
\end{align}
has the Laurent property. Any four adjacent iterates $k_n$, $l_n$, $k_{n+1}$, $l_{n+1}$ are pairwise coprime  Laurent polynomials in the ring
$
{\cal R}[k_0,k_1,l_0^{\pm 1},l_1^{\pm 1}],
$
where ${\cal R}=\Z [a_1,a_2,a_3,a_4,a_5]$. 
\end{proposition} 
\begin{proof}
This follows from Theorem 2 and a coprimality argument analogous to the proof of Proposition \ref{thm5}.
\end{proof}

It is straightforward to check that the equation  (\ref{fol}),  obtained from (\ref{um}) in \cite{HK}, corresponds to a degenerate case of (\ref{M12}).
\begin{proposition}
In the degenerate case $a_{1}=a_{3}=a_{4} =0$, $a_{2}=1$ and  $a_{5}=-\alpha$, the solution of (\ref{M12}) is given by
 $ k_{n}=e_{n+1}e_{n-2}$ and $l_{n}=e_{n}e_{n-1}$,
where $e_n$ satisfies equation (\ref{fol}).
\end{proposition}

\subsection{Growth of degrees of the additive QRT map}

Just as in the multiplicative case, the Laurent property means that the iterates of (\ref{M12}) can be factored as in 
(\ref{deform}), in terms of a general set of initial values ${\bf k}=(k_0,k_1,l_0,l_1)$, with initial denominator vectors (\ref{dvecs}), where
(for generic parameter values) $\rN_n$ and $\hat{\rN}_n$ are coprime. Observe that the system is also bilinear (homogeneous of degree 2), so the relations (\ref{homog}) hold, and the second equations
in (\ref{smsql}) and (\ref{M12}) are identical. Note also that, due to the minus sign in the first equation, the system (\ref{M12}) does not generate subtraction-free rational expressions in ${\bf k}$. Nevertheless, by considering the dependence on $k_0,k_1$, it is not hard to show that cancellations cannot occur in the highest degree terms appearing
in the numerator on the right-hand side, and arrive at the following 

\begin{lemma}\label{nosf}  
The denominator vectors satisfy the max-plus tropical version of
(\ref{M12}), namely
\beq\label{tropadd}
\begin{array}{rcl}
{\bf d}_{n+1}+{\bf e}_{n-1} & = & \max ({\bf d}_{n-1}+{\bf e}_{n+1},2{\bf d}_{n}, {\bf d}_{n}+{\bf e}_{n}, 2{\bf e}_{n}), \\
{\bf e}_{n+1}+{\bf e}_{n-1} & = & \max (2{\bf d}_{n}, {\bf d}_{n}+{\bf e}_{n}, 2{\bf e}_{n}).
\end{array}
\eeq
\end{lemma}
\begin{proof} By Proposition \ref{addlau}, the iterates of (\ref{M12}) can be written in the form (\ref{deform}), 
where now the (Laurent) monomial denominators are just monomials in $l_0,l_1$ for all $n\geq 2$. Then, by  explicitly 
considering how the numerators depend on $k_0,k_1$, 
$\rN_2({\bf k}) = -a_1 k_0 k_1^2 +\ldots$, $\hat{\rN}_2({\bf k}) = a_1 k_1^2 +\ldots$,  
and it can be shown by induction 
that 
for all $n\geq 2$, 
$\rN_n({\bf k})=$ monomial in $k_0,k_1$ of degree $\delta_n$ + 
lower order terms, and 
 $\hat{\rN}_n({\bf k})=$  monomial in $k_0,k_1$ of degree $\hat{\delta}_n$ + 
lower order terms, 
where $\delta_n,\hat{\delta}_n$ denote the respective total degrees of these numerators.
To perform the induction, it should be assumed also that $\delta_n>\hat{\delta}_n$ as part of the inductive 
hypothesis; clearly this holds for $n=2$, and $n=3$ is easily checked as well. 
Now suppose the hypothesis holds up to $n$, and consider the right-hand side
of the second equation in (\ref{M12}): in terms of $k_0,k_1$, the term of highest degree in the  numerator comes from $k_n^2$, and all other terms are of lower degree, so there can be no cancellation between the numerator and 
denominator (which only depends on $l_0,l_1$); thus, 
comparing with the numerator of the left-hand side, this implies 
\beq\label{del1}
\hat{\delta}_{n+1}+\hat{\delta}_{n-1} = 2\delta_n, 
\eeq 
and the numerator $\hat{\rN}_{n+1}({\bf k})$ is of the required form. 
Similarly,  on the  right-hand side of the first equation in  (\ref{M12}), 
the term with numerator 
of largest possible degree can only be  $k_{n-1}l_{n+1}$ or $k_n^2$, but the first term has 
degree $\delta_{n-1}+\hat{\delta}_{n+1}  =2\delta_n +\delta_{n-1}-\hat{\delta}_{n-1}>2\delta_n$,
using (\ref{del1}) and the inductive hypothesis, which implies that the numerator of  $k_{n-1}l_{n+1}$ 
is of largest degree; so again there can be no cancellation, 
the numerator ${\rN}_{n+1}({\bf k})$ is of the required form, and 
comparing with the left-hand side yields 
\beq\label{del2} 
\delta_{n+1}+\hat{\delta}_{n-1}=\hat{\delta}_{n+1}+\delta_{n-1}, 
\eeq 
and hence $\delta_{n+1}- \hat{\delta}_{n+1}=\delta_{n-1}-\hat{\delta}_{n-1}>0$, which 
gives the other part of the 
hypothesis.
\end{proof} 

From the form of the initial data (\ref{dvecs}), the solution of the ultradiscrete system (\ref{tropadd}) is written as
$$
{\bf d_n} = (d_n,d_{n-1},\tilde{d}_n,\tilde{d}_{n-1})^T, \qquad
{\bf e_n} =  (e_n,e_{n-1},\tilde{e}_n,\tilde{e}_{n-1})^T,
$$
in terms of two pairs of sequences $(d_n,e_n)$, $(\tilde{d}_n,\tilde{e}_{n})$, each of which satisfies the scalar version of
(\ref{tropadd}). Upon introducing the difference vector
$${\mathbf{\Delta}}_n={\bf d}_n-{\bf e}_n= (\Delta_n,\Delta_{n-1},\tilde{\Delta}_n,\tilde{\Delta}_{n-1})^T,$$
the  system is equivalent to
\beq\label{tropadds}
\begin{array}{rcl}
{\mathbf{\Delta}}_{n+1} & = & \max ({\mathbf{\Delta}}_{n-1}, 0), \\
{\bf e}_{n+1}-2{\bf e}_{n}+{\bf e}_{n-1} & = &  \max (2{\mathbf\Delta}_{n},0).
\end{array}
\eeq
For the first sequence pair $(d_n,e_n)$, the initial data give $\Delta_0=-1$, $\Delta_1=0$, which implies
$\Delta_n=0$ for $n\geq 1$, while
for the pair $(\tilde{d}_n,\tilde{e}_{n})$ with $\tilde{\Delta}_0=1$, $\tilde{\Delta}_1=0$,
the first equation above gives $\tilde{\Delta}_{\mod 2} =[0,1]$. For the second equation in (\ref{tropadds}), the
solution ${\bf e}_n$  is then found to be specified by
$$
  \tilde{e}_n = \frac{1}{2}n^2-\frac{3}{4}-\frac{(-1)^n}{4}
\quad \mathrm{and} \quad e_n = 0 \quad \forall n\geq 0
\implies d_n=0 \quad \forall n\geq 1,
$$
so $k_0,k_1$ never appear in the denominator of any Laurent polynomials, as is obvious from (\ref{M12}).

Finally, from (\ref{homog}) we see that, consistent with (\ref{del1}) and (\ref{del2}), 
$$\deg_{\bf k}   (\rN_n ({\bf k}))=
\Delta_n+\Delta_{n-1}+\tilde{\Delta}_n+\tilde{\Delta}_{n-1} +\deg_{\bf k}   (\hat{\rN}_n ({\bf k}))=
1+\deg_{\bf k}   (\hat{\rN}_n ({\bf k}))\quad \mathrm{for} \, n\geq 2$$
and $\deg_{\bf k}   (\hat{\rN}_n ({\bf k}))=1+e_n+e_{n-1}+\tilde{e}_n+\tilde{e}_{n-1}=n^2-n$.
By setting  ${\bf k}=(u_0,u_1,1,1)$ to find $u_n=k_n/l_n=\rN_n ({\bf u})/\hat{\rN}_n ({\bf u})$ for $n\geq 2$, and noting
that the total degrees of $\rN_n$ and $\hat{\rN}_n$ remain the same
after substitution, this yields the same quadratic expression for the degree growth 
as for (\ref{mul}).
\begin{theorem}\label{LL4}
As a rational function of the initial values $u_0,u_1$, the $n$th iterate $u_n$ of
 the additive QRT map (\ref{add}) has degree
$n^2-n+1$.
\end{theorem}

\section{The 12-parameter symmetric QRT map}
The symmetric QRT map \cite{QRT1,QRT2} is constructed as follows. We start with two symmetric $3\times 3$ matrices,
\beq\label{AB}
{\cal A}=\begin{pmatrix}
a_{00} & a_{01} & a_{02} \\
a_{01} & a_{11} & a_{12} \\
a_{02} & a_{12} & a_{22}
\end{pmatrix},
\quad
{\cal B}=\begin{pmatrix}
b_{00} & b_{01} & b_{02} \\
b_{01} & b_{11} & b_{12} \\
b_{02} & b_{12} & b_{22}
\end{pmatrix},
\eeq
and introduce the vectors
\[
{\bf V} (u,v)=\begin{pmatrix}
u^2 \\
uv \\
v^2
\end{pmatrix},
\qquad
{\bf f}(u)={\cal A}{\bf V}(u,1)\times {\cal B}{\bf V}(u,1),
\]
with the components of ${\bf f}$ denoted by $f^{(i)}$, $i=1,2,3$. The map of the plane defined by
\begin{equation} \label{QRT}
\varphi_{sym}:\, (u_{n-1},u_n)\mapsto(u_n,u_{n+1}),\quad\mathrm{ with}\quad
u_{n+1}=\frac{f^{(1)}(u_n)-u_{n-1}f^{(2)}(u_n)}{f^{(2)}(u_n)-u_{n-1}f^{(3)}(u_n)},
\end{equation}
is the general form of the symmetric QRT map, which
admits the invariant
\begin{equation} \label{INV}
J=\frac{{\bf V}(u_{n-1},1)^T{\cal A}{\bf V}(u_n,1)}{{\bf V} (u_{n-1},1)^T{\cal B}{\bf V}(u_n,1)}.
\end{equation}
\subsection{Recursive factorisation of the symmetric QRT map}
Let us homogenise the map $\varphi_{sym}$. Taking $u_n=p_n/q_n$ gives
\[
\frac{p_{n+1}}{q_{n+1}}=\frac{q_{n-1}f^{(1)}(\frac{p_n}{q_n})-p_{n-1}f^{(2)}(\frac{p_n}{q_n})}
{q_{n-1}f^{(2)}(\frac{p_n}{q_n})-p_{n-1}f^{(3)}(\frac{p_n}{q_n})},
\]
from which we obtain the polynomial system
\begin{equation} \label{QRTH}
p_{n+1}=q_{n-1}F^{(1)}_n-p_{n-1}F^{(2)}_n,\quad
q_{n+1}=q_{n-1}F^{(2)}_n-p_{n-1}F^{(3)}_n,
\end{equation}
with $F^{(i)}_n=q_n^4f^{(i)}(\frac{p_n}{q_n})$. 
Slightly more explicitly, in terms of the vectors ${\bf A}_n:={\cal A}{\bf V}(p_n,q_n)$, ${\bf B}_n:={\cal B}{\bf V}(p_n,q_n)$, with components denoted $A^{(i)}_n$, $ B^{(i)}_n$ respectively,
we set ${\bf F}_n={\bf A}_n\times {\bf B}_n$,
and then
the system (\ref{QRTH}) can be written as
\begin{equation} \label{BB1}
\begin{split}
p_{n+1}&=B^{(3)}_n(p_{n-1}A^{(1)}_n+q_{n-1}A^{(2)}_n)-A^{(3)}_n(p_{n-1}B^{(1)}_n+q_{n-1}B^{(2)}_n),\\
q_{n+1}&=B^{(1)}_n(p_{n-1}A^{(2)}_n+q_{n-1}A^{(3)}_n)-A^{(1)}_n(p_{n-1}B^{(2)}_n+q_{n-1}B^{(3)}_n),
\end{split}
\end{equation}
which is  equivalent to the vector equation
\begin{equation} \label{BB2}
\left(q_{n-1}q_{n+1},-p_{n-1}q_{n+1}-q_{n-1}p_{n+1},p_{n-1}p_{n+1}\right)^T={\bf W}_n,
\end{equation}
where ${\bf W}_n={\bf F}_n\times {\bf V}_{n-1}$, ${\bf V}_n={\bf V}(p_n,q_n)$.
The  map (\ref{BB1}) generates polynomials when iterated forwards, but not backwards, as 
the inverse does not have the Laurent property (this is analogous to the fact that a generic polynomial map does not 
have a polynomial inverse); 
it leaves the ratio of polynomials
$N_n/D_n$ invariant, where
$$ N_n ={\bf V}_n\cdot{\bf A}_{n-1} ={\bf V}_{n-1}\cdot{\bf  A}_n, \quad
D_n ={\bf  V}_n\cdot{\bf  B}_{n-1} = {\bf V}_{n-1}\cdot{\bf  B}_n,
$$
since
the matrices $\cal A$, $\cal B$ are symmetric. The invariance implies that
$N_n$ divides $N_{n+1}$ and $D_n$ divides $D_{n+1}$. We
can describe the  factorisations as follows.
\begin{lemma} \label{L15}
We have $N_{n+1} = N_nQ_n$ and $D_{n+1}=D_nQ_n$ where $Q_n=(F^{(2)}_n)^2-F^{(1)}_nF^{(3)}_n$.
\end{lemma}
\begin{proof}
Using  ${\bf F}={\bf A}\times {\bf B}$, 
hence ${\bf F}\cdot{\bf A}=0$,
${\bf W}={\bf F}\times{\bf V}=({\bf A}\cdot{\bf V}){\bf B}-({\bf B}\cdot {\bf V}){\bf A}$, we get
\begin{align*}
N_{n+1}&={\bf V}_{n+1}\cdot{\bf A}_n\\
&=p_{n+1}(q_{n-1}F^{(1)}_n-p_{n-1}F^{(2)}_n)A^{(1)}_n+(q_{n-1}F^{(2)}_n-p_{n-1}F^{(3)}_n) (p_{n+1} A_n^{(2)}+q_{n+1}A_n^{(3)}) \\
&=W^{(3)}(-F^{(2)}_nA^{(1)}_n-F^{(3)}_nA^{(2)}_n) + W^{(2)}F^{(3)}_nA^{(3)}_n + W^{(1)}F^{(2)}_nA^{(3)}_n\\
&={\bf A}_n\cdot{\bf V}_{n-1} \Big(F_2^n(A^{(3)}_nB^{(1)}_n-A^{(1)}_nB^{(3)}_n)-F^{(3)}_n(A^{(2)}_nB^{(3)}_n-A^{(3)}_nB^{(2)}_n)\Big)
=N_nQ_n.
\end{align*}
\end{proof}
Iterating the map one more time, the quotient $Q_n$ arises as a common
divisor of $p_{n+2}$ and $q_{n+2}$. This can be seen from
\begin{eqnarray}
p_{n+2}&=&B^{(3)}_{n+1}(p_{n}A^{(1)}_{n+1}+q_{n}A^{(2)}_{n+1})-A^{(3)}_{n+1}(p_{n}B^{(1)}_{n+1}+q_{n}B^{(2)}_{n+1})\notag\\
&=&\frac{B^{(3)}_{n+1}(p_{n}^2A^{(1)}_{n+1}+p_{n}q_{n}A^{(2)}_{n+1}+p_{n}^2A^{(3)}_{n+1})}{p_{n}}-\frac{A^{(3)}_{n+1}(p_{n}^2B^{(1)}_{n+1}+p_{n}q_{n}B^{(2)}_{n+1}+p_{n}^2B^{(3)}_{n+1})}{p_{n}}\notag\\
&=&\frac{B^{(3)}_{n+1}N_{n+1}-A^{(3)}_{n+1}D_{n+1}}{p_{n}}
=
Q_n \left(\frac{B^{(3)}_{n+1}N_n-A^{(3)}_{n+1}D_n}{p_{n}}\right). \label{pnp2}
\end{eqnarray}
The second term in (\ref{pnp2}) is polynomial, which can be seen directly from the factorisation
\[
B^{(3)}_{n+1}A^{(3)}_{n-1}-A^{(3)}_{n+1}B^{(3)}_{n-1}=(p_{n-1}q_{n+1}-p_{n+1}q_{n-1})({\bf F}_n\times{\bf V}_{n-1})\cdot{\bf P},
\]
where ${\bf P}=(a_{02},a_{12},a_{22})^T\times(b_{02},b_{12},b_{22})^T$,  and the observation that ${\bf F}_n\equiv{\bf P}q_n^4\mod p_n$.
For $q_{n+2}$ we find the similar expression
\begin{equation} \label{qnp2}
q_{n+2}=Q_n \left(\frac{B^{(1)}_{n+1}N_n-A^{(1)}_{n+1}D_n}{q_{n}}\right).
\end{equation}

\begin{theorem} \label{QdQ}
The polynomial $Q_n$ is a divisor of $Q_{n+1}$.
\end{theorem}
\begin{proof}
The proof is by direct computation. We write
\begin{equation} \label{detm}
Q_{n}=\sum_{i=0}^8 d_ip_{n}^{8-i}q_{n}^i = Q(p_{n},q_{n}).
\end{equation}
Considering the three components of ${\bf F}_n$ as variables,
we substitute equation (\ref{QRTH}) into $Q_{n+1}$
and reduce the result modulo $Q_n=(F^{(2)}_n)^2-F^{(1)}_nF^{(3)}_n$ using a total
degree ordering. We denote the coefficients of the resulting polynomial in $p_{n-1},q_{n-1}$
by $e_i$, so that
$
Q_{n+1}\mid_{\text{Eq.} (\ref{QRTH})} \equiv \sum_{i=0}^8 e_ip_{n-1}^{8-i}q_{n-1}^i \mod Q_n$.
There appear to be two non-trivial common factors, namely
$
X=\gcd(e_0,e_2,e_4,e_6,e_8)$, $Y=\gcd(e_1,e_3,e_5,e_7)$,
and we have $Q_n \mid Y- XF^n_2$. Thus it suffices to establish that $Q_n$
is a divisor of $X$. Curiously, we have the following expression, modulo $Q_n$:
$
X(F_n^{(1)})^4 \equiv \sum_{i=0}^8 d_i(F_n^{(1)})^{8-i}(F^{(2)}_n)^i= Q(F^{(1)}_n,F^{(2)}_n)$.
Finally, with computer algebra (e.g. Maple) it can be  verified that $Q_n$ divides $Q(F^{(1)}_n,F^{(2)}_n)$.
\end{proof}
From (\ref{BB1}), an ultradiscrete system of recurrences for a lower bound on the multiplicities is given as follows:
\begin{align}\label{uds}
\m^p_{n+1}=\m^q_{n+1}=\min_{r\in\{p,q\},i+j=4,i,j\geq 0} \left(\m^r_{n-1}+i\m^p_n+j\m_n^q\right).
\end{align}
So a lower bound on the growth of the multiplicity of divisors is 
$\m_{n+1}^{p}=\m_{n+1}^{q}=\m_{n+1}$ with
\begin{align}\label{pet1}
\m_{n+1}=4\m_{n}+\m_{n-1},
\end{align} where $\m_{0}=0$ and $\m_{1}=1$. We are interested in primitive divisors. Therefore, according to Theorem \ref{QdQ} we want to divide $Q_n$ by $Q_{n-1}$. But there is much more to divide out. As $Q_{n-2}$ divides $\gcd(p_n,q_n)$ and $Q_n$ is homogeneous in $p_n,q_n$ of degree 8, $Q_n$ is also divisible by $(Q_{n-2})^8$, and by $(Q_{n-3})^{32}$, and so on. We can recursively define a polynomial $r_n$ by $r_1=1$ and
\[
Q_n=\prod_{k=2}^{n-2} r_{n-k}^{8\m_{k-1}} r_{n-1} r_{n} \implies r_n=\prod_{j=0}^{n-2} (Q_{n-j})^{1-2j^2}.
\]
In terms of $r_n$, a common divisor of $p_n$ and $q_n$ is given by
\begin{equation} \label{cook} g_n=\prod_{k=2}^{n-2} r_k^{\m_{n-k-1}} \implies
g_{n+1}=r_{n-1}g_{n-1}g_{n}^4.
\end{equation}
As every divisor is a common divisor, we define the quotients $s_n,t_n$ 
by
\begin{align}\label{G3}
p_{n}=g_{n}s_{n} \quad \text{ and }\quad q_{n}=g_{n}t_{n}.
\end{align}
To find a closed system of equations we need to involve $r_n$
and find out how it relates to $s_n,t_n$. Using
$Q_n=Q(p_n,q_n)=g_n^8Q(s_n,t_n)$ it can be
verified that $r_nr_{n-1}=Q(s_n,t_n)$.
Substituting (\ref{G3}) into (\ref{QRTH}), and using (\ref{cook}), we arrive at the 
system 
 \begin{align}\label{KP1}
\begin{array}{ll}
r_{n}r_{n-1}&=Q(s_{n},t_{n}), \\
s_{n+1}r_{n-1}&=t_{n-1}F^{(1)}(s_{n},t_{n})-s_{n-1}F^{(2)}(s_{n},t_{n}),\\
t_{n+1}r_{n-1}&=t_{n-1}F^{(2)}(s_{n},t_{n})-s_{n-1}F^{(3)}(s_{n},t_{n}),
\end{array}
\end{align} with ${\bf F}(s_{n},t_{n})=g_n^{-4}{\bf F}_n$. Up to shifting indices,  initial values can be chosen as
$s_{0}=u_{0}$, $s_{1}=u_{1}$, $t_0=t_1=r_{0}=1$.
The ratio $u_n=\frac{s_n}{t_n}$ is an iterate of the symmetric QRT map (\ref{QRT}).
\begin{theorem}\label{LP}
The system (\ref{KP1}) has the Laurent property: 
$r_n,s_n,t_n\in\hat{\cal R}:={\cal R}[r_0^{\pm 1}, s_0, s_1, t_0, t_1]$, where
$\cal R$ is the ring of polynomials  in the parameters $a_{ij},b_{ij}$ over $\Z$.
\end{theorem}
\begin{proof} 
Upon noting that $Q(s,t)$ and $F^{(j)}(s,t)$ are  homogeneous with weights 8 and 4 respectively, we see that the system
(\ref{KP1}) is weighted homogeneous, where $r_n$ has degree 4 and $s_n,t_n$ have degree 1.  
This allows us to show by induction that  the iterates can be written in terms of ${\bf r}=(r_0,s_0,s_1,t_0,t_1)$
in a similar fashion  to (\ref{deform}), 
as
\beq \label{ndrst} 
 r_n = \frac{\rN^*_n (s_0,s_1,t_0,t_1)}{{r_0}^{{c}_n}}, \qquad
s_n = \frac{\rN_n (s_0,s_1,t_0,t_1)}{{r_0}^{{d}_n}}, \qquad
t_n = \frac{\hat{\rN}_n (s_0,s_1,t_0,t_1)}{{r_0}^{{e}_n}}, 
\eeq 
where $\rN^*_n$, $\rN_n$, $\hat{\rN}_n$ are homogeneous polynomials in $s_0,s_1,t_0,t_1$. 
The system (\ref{KP1}) only involves division by $r_{n-1}$ at each iteration, and since it is of second order in $s_n,t_n$ and only first order in $r_n$, it is slightly more general than the conditions in Theorem 2, but   
Hickerson's method still extends to this  situation. Clearly the 5  initial values as well as $r_1,s_2,t_2$ belong to $\hat{\cal R}$, 
while from Theorem \ref{QdQ} it follows that $Q(s_1,t_1)| Q(s_2,t_2)$, which implies $r_2\in \hat{\cal R}$, and the same computations that yield (\ref{pnp2}) and (\ref{qnp2}) also give $s_3,t_3\in \hat{\cal R}$. It can also be verified with 
computer algebra that $r_1$ and $r_2$ are coprime; it is sufficient to check that this is so for some particular numerical  choice of coefficients $a_{ij},b_{ij}$, in which case it must hold when the coefficients are variables. 
Then from the inductive hypothesis we can write the new iterate produced by the first equation in (\ref{KP1}) as 
$$
r_n = \frac{\rN^*_{n-1} (s_1,s_2,t_1,t_2)}{{r_1}^{{c}_{n-1}}} = \frac{\rN^*_{n-2} (s_2,s_3,t_2,t_3)}{{r_2}^{{c}_{n-2}}}.
$$ 
By 
substituting for $r_1,r_2,s_2,s_3,t_2,t_3$ as Laurent polynomials in $\hat{\cal R}$ and using $(r_1,r_2)=1$, it follows from the equality of the latter two expressions  above that $r_n\in\hat{\cal R}$, and the same argument shows that 
$s_{n+1},t_{n+1}\in\hat{\cal R}$. 
\end{proof}
In an appendix, we also prove the following result. 
\begin{lemma}\label{coprime} 
The Laurent polynomials $s_n$ and $t_n$ generated by (\ref{KP1}) are coprime in $\hat{\cal R}$ for all $n\geq 0$. 
\end{lemma} 

As before, no cancellations occur in the numerators when the Laurent polynomials are substituted into the right-hand sides of the system (\ref{KP1}), so we can immediately write down the ultradiscrete system for the denominator 
exponents $c_n,d_n,e_n$ in (\ref{ndrst}), 
that is
$$
\begin{array}{rcl}
{c}_n+{c}_{n-1} & = & \max \Big((8-i){d}_n +i{e}_n\Big)_{i=0,\ldots,8}, \\
 {d}_{n+1}+{c}_{n-1} & = & \max ({d}_{n-1},{e}_{n-1}) + M_n, \\
 {e}_{n+1}+{c}_{n-1} & = & \max ({d}_{n-1},{e}_{n-1}) + M_n,
\quad M_n:=  \max \Big((4-i){d}_n +i{e}_n\Big)_{i=0,\ldots,4}
\end{array}
$$
Subtracting the last two equations above implies that ${d}_n = {e}_n$ for all $n\geq 2$, and hence
$$
{c}_n+{c}_{n-1} =8{d}_n, \qquad
{d}_{n+1}+{c}_{n-1} =4{d}_n + {d}_{n-1} \implies ({\cal S}-1)^3 {d_n}=0.
$$
Therefore ${d}_n$ and ${e}_n$ both grow quadratically with $n$, as do the degrees of 
$\rN_n$, $\hat{\rN}_n$, and (by Lemma \ref{coprime} and its proof - see appendix) 
these are coprime and their degrees remain the same after substituting ${\bf r}=(1,u_0,u_1,1,1)$ into
$u_n =\rN_n/\hat{\rN}_n$. This yields 
an exact 
formula for the degree growth of  (\ref{QRT}). To be precise, for $n\geq 2$ we find
${ d}_n ={ e}_n={(n^2-n)}/{2}
\implies
\deg_{\bf u}(\rN_n)=\deg_{\bf u}(\hat{\rN}_n)=\deg_{\bf u}(u_n)=2n^2-2n+1$, using the fact that 
$\deg_{\bf r}(\rN_n)=4d_n+1$ by the weighted homogeneity of (\ref{KP1}). 
\begin{theorem}\label{LL5}
As a rational function of the initial values $u_0,u_1$, the $n$th iterate $u_n$ of
 the symmetric QRT map (\ref{QRT}) has degree $2n^2-2n+1$.
\end{theorem}

\subsection{Degeneration to the Laurentified multiplicative/additive QRT maps}
In this subsection, we show how the Laurentified multiplicative and additive QRT maps arise as special cases of the three-component system (\ref{KP1}).

\begin{theorem}
In the degenerate case $b_{11}=1$ and all other  $b_{ij}=0$, the polynomials $r_{n}$, $s_{n}$ and $t_{n}$ can be expressed in terms of the solution to (\ref{smsql}) as
\beq\label{H15}r_n = k_{n+1}l_{n+1}k_{n}l_{n}, \qquad
s_n=k_n, \qquad t_n=l_n.
\eeq
\end{theorem}
\begin{proof}
By substituting $b_{11}=1$ and no  other non-zero $b_{ij}$ in (\ref{KP1}), we find
\beq r_{n}r_{n-1}=s_{n}^2t_{n}^{2}R_{n}S_{n},\label{KP9}\quad
s_{n+1}r_{n-1}=t_{n-1}t_{n}s_{n}S_{n},\quad 
t_{n+1}r_{n-1}=s_{n-1}s_{n}t_{n}R_{n},
\eeq
where
$R_{n}={a_{{1}}s_{{n}}}^{2}+a_{{2}}s_{{n}}t_{{n}}+a_{{3}}{t_{{n}}}^{2}$ and $S_{n}={a_{{3}}s_{{n}}}^{2}+a_{{5}}s_{{n}}t_{{n}}+a_{{6}}{t_{{n}}}^{2}$. With the substitutions (\ref{H15}),
the second and third equations n (\ref{KP9}) correspond to the two equations (\ref{smsql}), while the first one is a consequence of them.
\end{proof}

\begin{theorem}
In the degenerate case $b_{22}=1$ and all other $b_{ij}=0$,  $r_{n}$, $s_{n}$  and  $t_{n}$ are expressed in terms of solutions to (\ref{M12}) as
\beq\label{QRTM} r_{n}=-l_{n}^{2}l_{n+1}^{2}, \qquad
s_{n}=k_{n},\qquad t_{n}=l_{n} .
\eeq
\end{theorem}
\begin{proof}
Setting  all  $b_{ij}=0$ apart from $b_{22}=1$ in (\ref{KP1}) gives
\beq\label{KP4}
r_{n}r_{n-1}=t_{n}^4R_{n}^2,\qquad
s_{n+1}r_{n-1}=t_{n}^2(s_{n-1}R_{n}+t_{n-1}  S_n), \qquad
t_{n+1}r_{n-1}=-t_{n-1}t_{n}^2R_{n},
\eeq
where $R_{n}=a_{1}s_{n}^{2}+a_{2}t_{n}s_{n}+a_{3}t_{n}^{2}$,
$S_n =a_{2}s_{n}^2+a_{4}s_{n}t_{n}+a_{5}t_{n}^2$.
The first and third equations above follow from the second equation in (\ref{M12}). The second equation in (\ref{KP4})
is a consequence of the first equation of the system (\ref{M12}).
\end{proof}

\section{Somos recurrences}
In this section we show that the components of the Laurent systems we have derived satisfy Somos-7 recurrence relations. We start with the multiplicative case, and from this we will obtain the additive case, before proceeding to obtain the Somos-7 relations corresponding  to the general symmetric QRT map as a corollary of a broader result for asymmetric QRT.

\subsection{The multiplicative case}
Let
\beq\label{iinv}
I=\frac {a_1u_0^2u_1^2+a_2u_0u_1(u_0+u_1)+a_3(u_0^2+u_1^2)+a_5(u_0+u_1)+a_6}{u_0u_1},
\eeq
be the invariant of (\ref{mul}) considered as a map on the plane of initial values
$(u_0,u_1)$. We will prove the following theorem.
\begin{theorem}\label{maint}
The variables $k_n,l_{n}$ in (\ref{smsql}) each satisfy the same Somos-7 recurrence,
\begin{equation} \label{s7}
x_{n+7}x_n=Ax_{n+6}x_{n+1}+Bx_{n+5}x_{n+2}+Cx_{n+4}x_{n+3},
\end{equation}
with coefficients given by
\begin{equation}\label{cofo}
A = -(\alpha^3+\alpha\beta\gamma+\beta^2), \qquad
B = (\beta^2-A)\alpha^2, \qquad
C = A(A-\beta^2),
\end{equation}
with
\begin{equation}
\label{abj}
\alpha = a_2a_5+a_3^2-a_1a_6+a_3 I, \quad \beta = (a_1a_5-2a_2a_3)a_5+a_2^2a_6+(a_1a_6-a_3^2)I ,
\quad \gamma=4a_3+I.
\end{equation}
\end{theorem}

The coefficients may be found using computer algebra as follows. If $x_n,y_n$ satisfy the same Somos-7 recurrence, then
\[ \label{mat}
\det \begin{pmatrix}
x_{5}x_{-2}&x_{4}x_{-1}&x_{3}x_{0}&x_{2}x_{1}\\
x_{6}x_{-1}&x_{5}x_{0}&x_{4}x_{1}&x_{3}x_{2}\\
y_{5}y_{-2}&y_{4}y_{-1}&y_{3}y_{0}&y_{2}y_{1}\\
y_{6}y_{-1}&y_{5}y_{0}&y_{4}y_{1}&y_{3}y_{2}
\end{pmatrix}=0
\]
for all initial values (where, in each row of the matrix, the indices can be shifted by an arbitrary amount). The coefficients of the Somos-7 recurrence are then found by calculating  a constant vector (independent of $n$) that spans the one-dimensional kernel  of the above matrix - see \cite{hones6} for an introduction to this method.

However, it is more instructive to prove the result by considering the form of the solution. As an added benefit,  this provides more succinct expressions for the coefficients. We first show that ${k}_n,{l}_n$ can be written in the form of products
\begin{equation}\label{prod}
{k}_n = \tau_n\,  \tau^*_n, \qquad {l}_n = \tilde{\tau}_n \,   \tilde{\tau}^*_n,
\end{equation}
where each of $\tau_n,  \tau^*_n,  \tilde{\tau}_n ,   \tilde{\tau}^*_n$, satisfy a Somos-5 recurrence (in fact,  the \textit{same} Somos-5 recurrence, with identical coefficients and a first integral taking the same value); then we use the fact that products of this kind provide special solutions of Somos-7 recurrences.

The general Somos-7 recurrence, of the form (\ref{s7}), arises as a reduction of the bilinear discrete BKP equation,
a partial difference equation which is also known as the cube recurrence \cite{FZ}. It is possible to obtain a general analytic solution in terms of genus two sigma functions, corresponding to a translation on  a two-dimensional Jacobian variety, which we intend to
present elsewhere. Nevertheless, Somos-7 also admits special  solutions given by products of elliptic sigma functions, which are described as follows.

\begin{proposition} Let  $\tau_n$ satisfy the Somos-5 recurrence
\begin{equation}
\label{s5}
\tau_{n+5} \tau_{n} = \alpha \tau_{n+4}\tau_{n+1}+\beta \tau_{n+3}\tau_{n+2},
\end{equation}
with coefficients $\alpha, \beta$,
and let $\tau_n^*$ satisfy the same recurrence but with coefficients  $\alpha^*, \beta^*$.
Then whenever the constraint
\begin{equation}\label{aform}
A = \frac{\psi_4\psi_2^*\psi_5^*}{\psi_2\psi_4^*} =  \frac{\psi_4^*\psi_2\psi_5}{\psi_2^*\psi_4}
\end{equation}
holds,
the product $k_n = \tau_n\,  \tau^*_n$ satisfies a Somos-7 recurrence of the form (\ref{s7}) with
$A$  given by (\ref{aform}), 
\begin{equation}\label{bcform}
B=\frac{\psi_3\psi_4\psi_3^*\psi_4^*}{\psi_2\psi_2^*}-A\psi_3\psi_3^*, \qquad
C=\psi_5\psi_5^* -A\, \frac{\psi_4\psi_4^*}{\psi_2\psi_2^*},
\end{equation}
where $\psi_n,\psi_n^*$ denotes the companion elliptic divisibility sequence associated with $\tau_n,\tau_n^*$ respectively.
\end{proposition}
\begin{proof} For later use, we record analytic formulae from \cite{honetams,honeswart} concerning the solution of (\ref{s5}), and its companion
elliptic divisibility sequence (EDS). From the proof of  Corollary 2.12 in \cite{honetams}, the general analytic solution of  (\ref{s5}) can be
written in the form
\begin{equation}\label{tau}
\tau_n=A_{\pm}
B_{\pm}^{[\frac{n}{2}]}
\mu^{[\frac{n}{2}]^2}\sigma (z_0+n\kappa), \qquad
\frac{B_+}{B_-}=-\mu ^{-1}=\sigma(\kappa)^4,
\end{equation}
with the $\pm$ signs selected according to the parity of $n$, 
where $\sigma (z)=\sigma (z;g_2,g_3)$ is the Weierstrass sigma
function corresponding to  the elliptic curve $\cal E$ in the $(x,y)$ plane given by
\begin{equation}\label{ellip}
{\cal E}: \qquad
y^2=4x^3-g_2x-g_3,
\end{equation}
and its companion
EDS  is given by
$
\psi_n = {\sigma (n\kappa)}/{\sigma (\kappa)^{n^2-1}}$.
The coefficients in (\ref{s5}) can be expressed in terms of the companion EDS, as
\begin{equation}\label{psiab}
\alpha = \psi_3, \qquad \beta = -\frac{\psi_4}{\psi_2},
\end{equation}
and the terms of the sequence $\tau_n$ satisfy a particular Somos-7 recurrence as well, namely
\begin{equation}\label{sps7}
\tau_{n+7}\tau_n = \frac{\psi_3\psi_4}{\psi_2} \, \tau_{n+5}\tau_{n+2}-\psi_5\, \tau_{n+4}\tau_{n+3},
\qquad \psi_5=\psi_2^3\psi_4-\psi_3^3
\end{equation}
(see also \cite{vdps}).
Upon substituting $\bar{k}_n = \tau_n\,  \tau^*_n$ into (\ref{s7}) and using (\ref{s5}), (\ref{sps7})  and the corresponding equations
(with asterisks) for $\tau_n^*$ to eliminate products with
shifts of width 5 and 7, one obtains a linear equation  in the products $XX^*$, $YY^*$, $XY^*$,$YX^*$, where $X=\tau_{n+5}\tau_{n+2}$,
$Y=\tau_{n+4}\tau_{n+3}$, and similarly for $X^*,Y^*$. Since this expression must  vanish for all $n$, it  is required that the coefficients of each of these
four products should be zero, leading to the two different equations for $A$
in (\ref{aform}), which have to be consistent, as well as the equations (\ref{bcform}) for $B,C$.
\end{proof}
\begin{corollary}\label{core}
If the sequences $\tau_n$, $\tau_n^*$ satisfy the Somos-5 recurrence (\ref{s5}) with the  same coefficients $\alpha,\beta$ and have the same value of
the invariant
\begin{equation}\label{tj}
\gamma=\frac{\tau_{n+1}\tau_{n+4}}{\tau_{n+2}\tau_{n+3}}+\frac{\tau_{n+1}\tau_{n+2}}{\tau_{n}\tau_{n+3}}+\alpha \left(\frac{\tau_{n+1}\tau_{n+4}}{\tau_{n+2}\tau_{n+3}}+\frac{\tau_{n}\tau_{n+3}}{\tau_{n+1}\tau_{n+2}}\right)
+\beta\frac{\tau_{n+2}^2}{\tau_{n}\tau_{n+4}},
\end{equation}
then  $k_n = \tau_n\,  \tau^*_n$  satisfies the Somos-7 recurrence (\ref{s7}) with coefficients
\begin{equation}\label{corco}
A=\psi_5, \qquad B=(D-A)\psi_3^2, \qquad C=A(A-D), \qquad \mathrm{where} \quad D=\left(\frac{\psi_4}{\psi_2}\right)^2.
\end{equation}
\end{corollary}
\noindent
\textit{Proof of Corollary:}
The fact that $\alpha,\beta,\gamma$ are the same for both sequences implies that they have the same companion EDS $\psi_n$. This means that the constraint (\ref{aform}) is trivially satisfied, and the expressions for the
coefficients  $A,B,C$ simplify dramatically.
\hfill $\qed$
\begin{theorem}\label{sol}
 The solution of the bilinear system  (\ref{smsql})  can be written in the form (\ref{prod}), where
$\tau_n,  \tau^*_n,  \tilde{\tau}_n ,   \tilde{\tau}^*_n$ all satisfy a Somos-5 recurrence with the same  coefficients $\alpha,\beta$ and the same value of
the invariant $\gamma$ given by (\ref{tj}).
\end{theorem}
\begin{proof} 
By Proposition 2.5 in \cite{honetams}, the solution of the multiplicative QRT map (\ref{mul}) can be written as
$u_n =f(z_0+n\kappa)$ where $f=f(z)$ is an elliptic function of its argument i.e. periodic with respect to the period lattice $\Lambda$ defined by an elliptic curve $\cal E$ as in (\ref{ellip}). The
function $f(z)$ provides a uniformization of the curve $\cal C$  in the $(u_0,u_1)$ plane defined by fixing the value of the invariant $I$ for the map, and this curve is isomorphic to $\cal E$. Also,
 projecting onto the first coordinate, $(u_0,u_1)\mapsto u_0$, defines a $2:1$
map ${\cal C} \rightarrow \mathbb{P}^1$, so the function $f$ is an elliptic function of order 2. Hence $f$ has two zeros and two poles in any period parallelogram for $\Lambda$, and so, by a standard result in the theory of elliptic functions (see e.g. \S 20$\cdot$53 in \cite{ww}), up to a 
shift in $z$,  
$$
f(z) =K \frac{\sigma(z-Z)\sigma(z+Z)}
{\sigma(z-P)\sigma(z+P)}$$
for some constants $Z,P,K$. Thus
\beq\label{urat}
u_n =\frac{k_n}{l_n} = K \frac{\sigma(z_n-Z)\sigma(z_n+Z)}{\sigma(z_n-P)\sigma(z_n+P)}, \qquad z_n=z_0+n\kappa,
\eeq
and it remains to check that by inserting suitable $n$-dependent prefactors as in (\ref{tau}) the numerator $k_n$ and denominator $l_n$ can each be written as  a
product of two Somos-5 sequences. To be more precise, the general solution of the bilinear system (\ref{smsql}) can be written as
\begin{equation}\label{msqlsoln}
k_n = \mathrm{a}_\pm \mathrm{b}_\pm^{[\frac{n}{2}]}\mu^{2[\frac{n}{2}]^2}\sigma(z_n-Z)\sigma(z_n+Z),
\qquad
l_n = \tilde{\mathrm{a}}_\pm \mathrm{b}_\pm^{[\frac{n}{2}]}\mu^{2[\frac{n}{2}]^2}\sigma(z_n-P)\sigma(z_n+P),
\end{equation}
where the $\pm$ signs are chosen with the parity of $n$, and the
prefactors are taken to satisfy the relations
$$ \mu=-\sigma (\kappa)^{-4}, \qquad
\frac{ \mathrm{a}_+}{\tilde{\mathrm{a}}_+}=\frac{ \mathrm{a}_-}{\tilde{\mathrm{a}}_-}=K, \qquad
\frac{\mathrm{b}_-}{\mathrm{b}_+}=\mu^2, \qquad
\mathrm{b}_+\mathrm{b}_-=\left(\frac{ \mathrm{a}_-}{{\mathrm{a}}_+}\right)^4.
$$
Subject to the above,
the solution (\ref{msqlsoln}) is completely determined by the 9 parameters
$ \mathrm{a}_+$, $ \mathrm{a}_-$, $\tilde{\mathrm{a}}_+$, $ z_0$, $\kappa$, $ Z$, $P$, $g_2$, $g_3$, which fix the 4 initial values
and 5 coefficients $a_1,a_2,a_3,a_5,a_6$ in (\ref{smsql}). Then from the formula for $k_n$ in (\ref{msqlsoln}) there is a factorisation $k_n=\tau_n\,\tau_n^*$, with $\tau_n$ given by (\ref{tau}) with $z_0$ replaced by $z_0-Z$, and
$$
\tau_n^*=A^*_{\pm}
B_{\pm}^{*\, [\frac{n}{2}]}
\mu^{[\frac{n}{2}]^2}\sigma (z_0+Z+n\kappa), \quad
\frac{B^*_+}{B^*_-}=-\mu ^{-1}=\sigma(\kappa)^4, \quad
\mathrm{where} \,\,
A_{\pm} A^*_{\pm} =\mathrm{a}_\pm, \, B_{\pm} B^*_{\pm} =\mathrm{b}_\pm,
$$
so that $\tau_n$ and $\tau_n^*$ satisfy the same Somos-5 recurrence with the same value of the invariant $\gamma$. By
making an analogous factorisation, the right-hand side of the formula for $l_n$ in (\ref{msqlsoln})  can be written as
a product $  \tilde{\tau}_n \,   \tilde{\tau}^*_n$ for another pair of such  Somos-5 sequences.
\end{proof}
It is worth explaining the nature of the factorisation (\ref{prod}) in more detail. First of all, by 
Proposition \ref{thm5},  each of  the terms $k_n,l_n$ is a
Laurent polynomial in the ring
$
{\cal R}[k_0^{\pm 1},k_1^{\pm 1},l_0^{\pm 1},l_1^{\pm 1}],
$  
and can be factored as a product of a polynomial with a Laurent monomial, but the factorisation 
(\ref{prod}) is not of this 
kind: it leads to $\tau_n,\tau_n^*$, etc. which are algebraic functions of the coefficients and initial data for  
 (\ref{smsql}). Secondly, the factorisation (\ref{prod}) is not unique: there is a 6-parameter family 
of gauge transformations that preserve it, given by 
\beq\label{gauge}
\tau_n\to \overline{A}_{\pm}\overline{B}^n\tau_n, \quad
\tau_n^*\to (\overline{A}_{\pm})^{-1}\overline{B}^{-n}\tau_n^*,\quad
\tilde{\tau}_n\to \hat{A}_{\pm}\hat{B}^n\tilde{\tau}_n, \quad 
\tilde{\tau}_n^*\to (\hat{A}_{\pm})^{-1}\hat{B}^{-n}\tilde{\tau}_n^*, 
\eeq 
where $\overline{A}_{\pm},\hat{A}_{\pm}$ vary with the parity of $n$,  for arbitrary non-zero 
$\overline{A}_{+}, \overline{A}_{-}, \hat{A}_{+}, \hat{A}_-,\overline{B}, \hat{B}$. 
The above theorem 
guarantees that Somos-5 sequences ${\tau}_n$,  ${\tau}_n^*$,  $\tilde{\tau}_n$,  $\tilde{\tau}_n^*$ 
exist such that (\ref{prod}) holds. However, given the initial data $k_0,l_0,k_1,l_1$ and coefficients $a_1,a_3,\ldots,a_6$ for  (\ref{smsql}), although the coefficients $\alpha,\beta$  are given by the polynomial expressions (\ref{abj}), 
it is necessary to solve algebraic equations in order to find the four sets of initial data for  Somos-5, i.e. 
$\tau_0,\tau_1,\ldots, \tau_4$ for the first sequence, and so on. 
Without loss of generality, by exploiting the gauge symmetry (\ref{gauge}),  the 6 values 
$\tau_0^*,\tau_1^*,\tau_2^*,\tilde{\tau}_0^*,\tilde{\tau}_1^*,\tilde{\tau}_2^*$ can all 
be fixed to 1, so that $\tau_0=k_0$, $\tau_1=k_1$, $\tau_2=k_2$ and 
 $\tilde{\tau}_0=l_0$, $\tilde{\tau}_1=l_1$, $\tilde{\tau}_2=l_2$. By fixing $\gamma$ as in (\ref{abj}), 
the formula (\ref{tj}) determines $\tau_4$ as an algebraic function of $\tau_0,\tau_1,\tau_2,\tau_3$ and the coefficients and initial data 
for (\ref{smsql}), by solving a quadratic equation, and similarly for $\tau_4^*,\tilde{\tau}_4, \tilde{\tau}_4^*$. Thus, 
in order to have four complete sets of initial data for Somos-5, 
it remains to determine the four quantities $\tau_3,\tau_3^*,\tilde{\tau}_3, \tilde{\tau}_3^*$. 
The product $\tau_3\tau_3^*=k_3$ is a known (Laurent polynomial) function of   $k_0,l_0,k_1,l_1$ and $a_1,a_3,\ldots,a_6$, and 
similarly for $\tilde{\tau}_3\tilde{\tau}_3^*=l_3$, but two more relations in are needed to obtain the remaining 
four quantities. Then, by using (\ref{s5}) to write 
$$ 
k_5 = \tau_5\tau_5^*
= \frac{(\alpha\tau_1\tau_4+\beta\tau_2\tau_3)( \alpha\tau_1^*\tau_4^*+\beta\tau_2^*\tau_3^*)}
{\tau_0\tau_0^*}, 
$$ 
the left-hand side is determined by the coefficients/initial data for  (\ref{smsql}), while the ratio 
on the right-hand side is a rational function of these together with the Somos-5 initial data, and 
similarly for $l_5= \tilde{\tau}_5\tilde{\tau}_5^*$, thus  providing the necessary two additional 
relations, which are best solved using resultants. 

\noindent
\textit{Proof of Theorem \ref{maint}:}
By Theorem \ref{sol}, both $k_n$ and $l_n$ are products of Somos-5 sequences, so by Corollary \ref{core}
they satisfy the same Somos-7 recurrence with coefficients given by (\ref{corco}). It remains to relate the quantities
$\psi_2,\psi_3,\psi_4$ that define the companion EDS to the parameters $a_j$, $j=1,2,3,5,6$ and the value of the integral $I$ for the QRT map (\ref{mul}). This is achieved by using the formulae for the EDS in \cite{honeswart}, which give the identity
$
\psi_2^4= \beta +\alpha \gamma
$ in addition to (\ref{psiab}), and noting that fixing the value of the invariant $\gamma$ defines a curve of genus one
 in the  plane with coordinates
$(x,y)=(\tau_n\tau_{n+3}/(\tau_{n+1}\tau_{n+2}),(\tau_{n+1}\tau_{n+4}/(\tau_{n+2}\tau_{n+3}) )$, that is
\beq \label{s5curve}
(x+y)xy+\alpha (x+y)-\gamma xy+\beta=0.
\eeq
The latter curve is isomorphic to the curve $\cal C$, and by the use of an algebra package such as {\tt algcurves} in Maple one can  relate $\alpha,\beta,\gamma$ to the coefficients appearing in $\cal C$; this leads to the relations (\ref{abj}).
\hfill $\qed$

\subsection{The additive case}
The quantity
\beq\label{jinv}
J=a_1u_0^2u_1^2+a_2u_0u_1(u_0+u_1)+a_3(u_0^2+u_1^2)+a_4u_0u_1+a_5(u_0+u_1),
\eeq
is an invariant of (\ref{add}) considered as a map on the plane of initial values
$(u_0,u_1)$. 
There is a well known 
link between the additive form (\ref{add}) of the QRT map and the multiplicative one (\ref{mul}).
\begin{lemma} \label{ij}
Suppose that an orbit of (\ref{mul}) starting from the initial point $(u_0,u_1)$ is such that the value of the  invariant is
$I=-a_4$. Then this coincides with the orbit of (\ref{add}) starting from the same  initial point with the invariant taking the value $J=-a_6$.
\end{lemma}
\begin{proof}
This follows from (\ref{iinv}), expressing 
$a_6$ in terms of 
$I$, 
substituting into the formula for the corresponding map, then comparing with 
$a_4$ in terms of $J$ obtained from
(\ref{jinv}).
\end{proof}
\begin{theorem}
The variables $k_n,l_n$ in (\ref{M12}) each satisfy the same Somos-7 recurrence, of the form
\begin{equation} \label{s7a}
x_{n+7}x_n=Ax_{n+6}x_{n+1}+Bx_{n+5}x_{n+2}+Cx_{n+4}x_{n+3}
\end{equation}
with  coefficients given by (\ref{cofo}), but with
$$
\alpha = a_1 J   + a_2a_5 +a_3^2-a_3 a_4, \quad \beta = a_1a_5^2-2a_2a_3a_5+a_3^2a_4  +(a_1a_4-a_2^2)J,
\quad \gamma=4a_3-a_4.
$$
\end{theorem}
\begin{proof}
It is enough to observe that, by Lemma \ref{ij},  any orbit of the multiplicative map  corresponds to an  orbit of the additive version, and by identifying ${k}_n, {l}_n$ in  (\ref{smsql}) with $k_n,l_n$ in  (\ref{M12}) the result follows.
\end{proof}

As already noted, the DTKQ-2 equation (\ref{um}) is a special case of the additive QRT map, with
$$
J= (u_nu_{n-1} - \al )(u_n+u_{n-1})
$$
being a first integral  in this case.
Moreover,  (\ref{fol}) arises from  (\ref{um})  by setting
$$
u_n = \frac{e_ne_{n+3}}{e_{n+1}e_{n+2}}.
$$
To see that how this is a degenerate case of the preceding result on Somos-7 recurrences, observe that, by  Lemma \ref{ij},
the iterates of  (\ref{um}) also satisfy the multiplicative QRT map
$$
u_{n+1}u_{n-1} = -\frac{(\al u_n +J)}{u_n},
$$
which leads to the following result.
\begin{proposition}\label{dtkq2}  Any sequence $(e_n)$ generated by (\ref{fol})  also satisfies the Somos-5 relation
\beq\label{ds5}
e_{n+5}e_n+\al\, e_{n+4}e_{n+1}+J\, e_{n+3}e_{n+2}=0 .
\eeq
\end{proposition}
\begin{corollary}  For all $m,n\in\Z$ the following relation holds:
\begin{equation} \label{s74}
\left|\begin{matrix}
e_{n}e_{n+5} & e_{m}e_{m+5} \\
e_{n+2}e_{n+3} & e_{m+2}e_{m+3} \\
\end{matrix}
\right|+
\alpha \left|\begin{matrix}
e_{n+1}e_{n+4} & e_{m+1}e_{m+4} \\
e_{n+2}e_{n+3} & e_{m+2}e_{m+3} \\
\end{matrix}
\right|
=
0.
\end{equation}
\end{corollary}
\noindent{\it Proof of Corollary:}
This follows  immediately from (\ref{ds5}). By eliminating $e_{m+5}$ from the top right entry and expanding the two determinants, the left-hand side of (\ref{s74})  becomes
$$ \begin{array}{ll}
&
\left|\begin{array}{cc} e_ne_{n+5} &  - \al\, e_{m+4}e_{m+1}-J\, e_{m+3}e_{m+2}\\
e_{n+2}e_{n+3} &  e_{m+2} e_{m+3} \end{array}\right| + \al \,
\left|\begin{array}{cc} e_{n+1}e_{n+4} &  e_{m+1} e_{m+4} \\
e_{n+2}e_{n+3} &  e_{m+2} e_{m+3} \end{array}\right|
\\ =& e_{m+2}e_{m+3} \, (e_{n+5}e_n+\al\, e_{n+4}e_{n+1}+J\, e_{n+3}e_{n+2})=0
\end{array} $$
Thus apart from an overall $m$-dependent prefactor, which can be removed, for each $n$ the formula (\ref{s74}) just corresponds to the same Somos-5 relation (\ref{ds5}).
\hspace{2.2in}$\Box$

\begin{remark} The equation (\ref{s74}) is of degree four and depends on two indices $m,n$. This is reminiscent of Ward's defining relation
for an elliptic divisibility sequence \cite{ward}.
\end{remark}

\subsection{The asymmetric QRT map}
It turns out that Theorem \ref{maint} admits a straightforward generalization to both the general 12-parameter symmetric QRT map and the full 18-parameter asymmetric QRT map, with the result for the former being a special case of that for the latter. 
For convenience, we briefly outline the geometrical description of the general QRT map due to Tsuda \cite{tsuda}; for more details see Duistermaat's book \cite{duistermaat}. Starting from a pencil of biquadratic curves
\begin{equation}\label{pencil}
{\cal P}(u,v;\la ) \equiv \sum_{i,j=0,1,2} (a_{ij}+\la \, b_{ij})u^{2-i}v^{2-j}=0
\end{equation}
in the $(u,v)$ plane, labelled by $\la\in \mathbb{P}^1$, there are two natural birational involutions on each curve (for fixed $\la$),
called the horizontal/vertical switch, denoted $\iota_H$ and $\iota_V$ respectively, 
given by
\begin{equation}\label{switch}
\iota_H: \quad (u,v)\mapsto (u^\dagger,v), \qquad \iota_V: \quad (u,v)\mapsto (u,v^\dagger),
\end{equation}
where  $u^\dagger$ denotes the conjugate root of (\ref{pencil}) considered as a quadratic in $u$ (for fixed $v$), and similarly for $v^\dagger$. Note that the birationality of $\iota_H$ (and similarly that of $\iota_V$) can be seen
directly  by writing (\ref{pencil}) as
$$
{\cal P}\equiv  \rA(v,\lambda) u^2 + \rB(v,\lambda) u + \rC(v,\lambda) =0,
$$
and using either the formula for the sum or the product of the roots, i.e.
\begin{equation}\label{sumprod}
u^\dagger + u = -\frac{\rB(v,\lambda)}{\rA(v,\lambda)}, \qquad u^\dagger \,  u = \frac{\rC(v,\lambda)}{\rA(v,\lambda)}.
\end{equation}
Then the QRT map $\varphi$ is defined on each curve in the pencil to be the product of these two involutions,
\begin{equation} \label{qrt}
\varphi = \iota_V\cdot\iota_H,
\end{equation}
and this lifts to a birational map on the $(u,v)$ plane: to apply $\iota_H$ one can eliminate $\la$ from the pair of equations (\ref{sumprod}) to yield
a formula for $u^\dagger$ as a rational function of $u$ and $v$ alone, and similarly for the application of $\iota_V$. Moreover, solving (\ref{pencil}) for $\la$ gives
\begin{equation}\label{la}
\la = -\frac{\sum_{i,j} a_{ij}u^jv^k}{\sum_{i,j} b_{ij}u^jv^k} ,
\end{equation}
and a generic point $(u,v)$, as well as its orbit under the QRT map $\varphi$, belongs to a unique curve ${\cal C}={\cal C}(\la)$ with the corresponding value of $\la$ being
determined by the above ratio (the exception being the base points, for which both the numerator and denominator vanish).

In the special case where the coefficient matrices ${\cal A}=(a_{ij})$ and ${\cal B}=(b_{ij})$ are symmetric, each curve in the pencil admits the additional involution
$
\sigma: \, (u,v)\mapsto (v,u)$,
and the symmetric QRT map $\varphi_{sym}$ is defined as
$
\varphi_{sym}=\sigma\cdot \iota_H$;
 the vertical switch is conjugate to the horizontal, so $\iota_V=\sigma\cdot\iota_H\cdot\sigma$, and
$\varphi=(\varphi_{sym})^2$. With symmetric matrices as in 
(\ref{AB}), $\varphi_{sym}$
coincides with the map (\ref{QRT}) defined previously, and comparing  (\ref{INV}) with (\ref{la}) we see that $J=-\la$.

Considering the general asymmetric QRT map, we can start from an initial point $(u_0,v_0)$ and apply $\varphi$
repeatedly to obtain the sequence
\begin{equation}\label{qrtseq}
(u_n,v_n)=\varphi^n (u_0,v_0).
\end{equation}
We have the following result.

\begin{theorem}\label{genqrt}
Any orbit (\ref{qrtseq}) produced by the general QRT map can be written
as $u_n = k_n/l_n$, $v_n=p_n/q_n$,
where $k_n,l_n,p_n,q_n$ all satisfy the same Somos-7 recurrence of the form (\ref{s7}),
with the coefficients $A,B,C$ given by (\ref{cofo}) in terms of quantities $\alpha,\beta,\gamma$ which are determined precisely by the 18 parameters $a_{ij},b_{ij}$ and the value of the invariant $\la$ as in  (\ref{la}).
\end{theorem}
\begin{proof} This is similar 
to the proof of Theorem \ref{sol}, so rather than repeating this it is sufficient to explain the main things that are different.
 For generic values of
$a_{ij},b_{ij}$  and $\la$ the curve $\cal C$ defined by (\ref{pencil}) is smooth for $(u,v)\in\Pro^1 \times \Pro^1$ and
has genus one, so it is isomorphic to $\C /\Lambda$ where $\Lambda$ is the period lattice of an elliptic curve
$\cal E$ given in the Weierstrass form  (\ref{ellip}).  Each of the maps $(u,v)\mapsto u$ and $(u,v)\mapsto v$ defines a double cover of $\Pro^1$, with associated elliptic functions of order two, denoted  $f,g$  respectively, such that
$(u,v)=(f(z),g(z))$ provides the uniformization of $\cal C$. Now, up to an overall shift in $z$,  we may write
$$
f(z) =K \, \frac{\sigma(z-Z)\sigma(z+Z)}
{\sigma(z-P)\sigma(z+P)}, \qquad g(z) =K' \, \frac{\sigma(z+\delta-Z')\sigma(z+\delta+Z')}
{\sigma(z+\delta-P')\sigma(z+\delta+P')}$$
for some constants $\delta,Z,P,K,Z',P',K'$ which specify the poles and zeros of $f,g$. Then by a standard argument (see
e.g. Theorem 2.4 in \cite{tsuda} or the proof of Proposition 2.5 in \cite{honetams}), the composition of two involutions $\varphi = \iota_V\cdot \iota_H$ corresponds to a translation $z\mapsto z+\kappa$ on the torus $\C /\Lambda$, so 
the  iterates (\ref{qrtseq})  of the QRT map have an analytic expression given by  (\ref{urat}) and
\beq\label{vnform}
v_n =\frac{p_n}{q_n} = K'\, \frac{\sigma(z_n+\delta-Z')\sigma(z_n+\delta+Z')}{\sigma(z_n+\delta-P')\sigma(z_n+\delta+P')}, \qquad z_n=z_0+n\kappa.
\eeq
Thus once again we can write down a factorisation $k_n=\tau_n\tau_n^*$ into a product of Somos-5 sequences, and similarly for $l_n,p_n,q_n$. Therefore all of $k_n,l_n,p_n,q_n$  satisfy the same Somos-7 recurrence, whose coefficients can be determined in terms of the  parameters  $\alpha,\beta,\gamma$ for a cubic curve (\ref{s5curve}) that is isomorphic to $\cal C$.
(See e.g. Remark 2.2 in \cite{tsuda} for details.)
\end{proof}

\begin{remark}\label{DP6} 
The formulae for $k_n,l_n,p_n,q_n$ as a product of two tau functions/sigma functions, 
as in (\ref{prod}), (\ref{msqlsoln}), and in the proof above, have a counterpart in the non-autonomous setting, 
in the form of the bilinearization of the $q$-$P_{III}$ and $q$-$P_{VI}$ equations by Jimbo et al. \cite{dp6}.
\end{remark}

The analogue of the above
result for the map $\varphi_{sym}=\sigma\cdot\iota_H$  is simpler  because for a  symmetric biquadratic curve
the shift $\kappa$ for the original QRT map
$\varphi$  can be chosen so that
$g(z)=f(z+\kappa/2)$, and 
then we can identify $v_n=u_{n+1/2}$ in (\ref{qrtseq}). Lemma \ref{ij} generalizes to the full 12-parameter symmetric case, such that on every orbit of  (\ref{QRT}) one may write $u_n$ as a ratio of quantities which satisfy both bilinear systems
(\ref{smsql}) and (\ref{M12}), with parameters that are linear functions of $\la$.  By combining the above results on the multiplicative and additive
cases, we arrive at the following.
\begin{corollary} \label{newpq}
Consider the general symmetric QRT map (\ref{QRT}) with
$u_n=p_n/q_n$, where  the variables $p_n,q_n$ satisfy a rational, non-Laurent system which is a variant of   (\ref{QRTH}), namely
\begin{equation} \label{quintic}
p_{n+1}=\frac{q_{n-1}F^{(1)}_n-p_{n-1}F^{(2)}_n}{D_n},\quad
q_{n+1}=\frac{q_{n-1}F^{(2)}_n-p_{n-1}F^{(3)}_n}{D_n},
\end{equation}
with $F^{(i)}_n=q_n^4f^{(i)}(\frac{p_n}{q_n})$ and
$D_n={\bf V}(p_{n-1},q_{n-1})^T{\cal B}{\bf V}(p_n,q_n)$ as above.
Then $p_n$ and $q_n$ are both solutions of the same Somos-7 relation, namely
\beq
\label{s763}
p_{n+7}p_n = Ap_{n+6}p_{n+1} + Bp_{n+5}p_{n+2}+Cp_{n+4}p_{n+3},
\eeq
where the formulae (\ref{cofo}) and (\ref{abj}) for the coefficients $A,B,C$ are still valid if we set
$$
a_1=a_{00}-J b_{00}, \, a_2=a_{01}-J b_{01},\,  a_3=a_{02}- J b_{02}, \,
a_5=a_{12}-J b_{12}, \,
a_6=a_{22}-J b_{22}, \,  I=-a_{11}+J b_{11}
$$
in terms of the 12 parameters $a_{ij},b_{ij}$, $0\leq i\leq j\leq 2$ and first integral $J=-\la$, as in (\ref{INV}), for
$\varphi_{sym}$.
\end{corollary}
\begin{corollary}\label{K} The variables  $s_n, t_n$ in the system (\ref{KP1})  both satisfy the same Somos-7 relation, given by (\ref{s763}) but with the coefficients rescaled according to
$$
A \to K^{12}A, \qquad B \to K^{20} B, \qquad C \to K^{24}C,
$$
where on each orbit the constant  $K$ is given by
\beq \label{kpqst}
K^2=\frac{\rho_{n+1}\rho_{n-1}}{\rho_n^2}, \qquad \rho_n =\frac{s_n}{p_n}=\frac{t_n}{q_n},
\eeq
with $p_n$, $q_n$ being the corresponding solution to (\ref{quintic}).
\end{corollary}
\noindent \textit{Proof of Corollary \ref{K}:}
Since the solutions of (\ref{KP1}) and (\ref{quintic}) are both related to the same solution of
(\ref{QRT}) via $u_n=s_n/t_n=p_n/q_n$, we can introduce $\rho_n =s_n/p_n=t_n/q_n$. Taking the quotient of
each side of the second equation (\ref{KP1}) with each side of the first equation (\ref{quintic}), and doing the same
for the third equation (\ref{KP1}) with the second equation (\ref{quintic}), leads 
in both cases to 
\beq \label{rho1}
\rho_{n+1}r_{n-1}= D_n \rho_n^4\rho_{n-1},
\eeq
while eliminating $s_n,t_n$ from the first equation (\ref{KP1}) gives
\beq\label{rho2}
r_nr_{n-1} = \rho_n^8 Q(p_n,q_n).
\eeq
Then using (\ref{rho1}) 
in (\ref{rho2}) and further substituting
for $p_{n+1},q_{n+1}$ from (\ref{quintic}) gives
$$
\left(\frac{\rho_{n+1}\rho_{n-1}}{\rho_n^2}\right) \left(\frac{\rho_{n+2}\rho_{n}}{\rho_{n+1}^2}\right)^{-1}
=\frac{Q(p_n,q_n)}{D_nD_{n+1}}=1 \implies \frac{\rho_{n+1}\rho_{n-1}}{\rho_n^2}=\mathrm{const}
\implies \rho_n = \rho_0 \left(\frac{\rho_1}{K\rho_0}\right)^n K^{n^2}
$$
for some $K$, which verifies (\ref{kpqst}). The quadratic exponential form of the gauge transformation from $p_n$ to $s_n$ means that $s_n$ satisfies (\ref{s763}) but with coefficients rescaled as stated, and likewise for $t_n$. If the initial values are chosen so that
$r_0=1$, $s_0=p_0=u_0$, $s_1=p_1=u_1$ and $t_0=q_0=t_1=q_1=1$ then $\rho_0=\rho_1=1$ and hence
$K^2=\rho_2=D_1={{\bf V} (u_{0},1)^T{\cal B}{\bf V}(u_1,1)}$. \hspace{1.7in}
$\Box$

\section{Concluding remarks}

We have applied homogenisation and/or  recursive factorisation to  symmetric QRT maps, and have shown directly that this produces systems with the Laurent property. In the multiplicative case, the resulting system  (\ref{smsql}), appears to correspond to a pair of mutations in an LP algebra \cite{LP}, which
would give an
alternative way to verify the Laurent property. However, 
neither (\ref{M12}) nor the system (\ref{KP1}) obtained from the  general
symmetric QRT 
are  of the right form for successive  LP algebra mutations. 

Our results show that obtaining Laurent systems and using their ultradiscrete (or tropical) versions is a very efficient method for calculating the degree growth 
of maps. For the QRT maps, which preserve an elliptic fibration, 
the fact that the degree growth is quadratic is well known  in the context of algebraic entropy of maps 
with invariant curves \cite{bellon}. 
Earlier results on automorphisms of rational surfaces admitting an elliptic fibration were presented by 
Gitzatullin \cite{ratsurfaces}.  
The quadratic growth observed in QRT maps 
also 
fits into Diller and Favre's classification of bimeromorphic maps of surfaces in \cite{df}, 
and can be understood by making a sequence of blowing-up transformations 
which lifts the QRT map to an automorphism of a complex analytic surface, 
then considering the action on homology (see e.g. the discussion of geometric singularity confinement in chapter 3 of \cite{duistermaat}).  
However, this approach may become intractable for maps in dimension 3 and above, whereas the combination of   Laurentification and 
ultradiscretization extends to   higher dimensions in a straightforward manner. 

In future work we propose to consider the analogue of (\ref{KP1}) for the general asymmetric QRT map, 
as well as Laurent systems for higher-dimensional maps. As a starting point, it would be worth considering the 
bilinearization of the discrete Painlev\'{e} VI ($q$-$P_{VI}$) equation presented in \cite{dp6}, for which
the  
autonomous limit  is just the multiplicative  version of the asymmetric QRT map. 

\noindent {\bf Acknowledgements:}
This research was supported by the Australian Research Council, grant DP140100383,  
and La Trobe's Disciplinary Research Program Mathematical and Computing Sciences. 
ANWH is supported by Fellowship EP/M004333/1 from the Engineering \& Physical Sciences Research Council.
We are grateful to the referees for their reports, and especially to one of the editors, Manuel Kauers, for his detailed comments on the original version of the paper.

\section*{Appendix} 

Here we provide an inductive proof of Lemma \ref{coprime}, by making repeated use of two  facts: (i) for a pair of  polynomials $f,g\in\mathbb{K}[x_1,\ldots,x_m]$ in $m$ variables  over a field $\mathbb{K}$ with positive degree in $x_1$, 
there are polynomials $A,B$ such that 
$
Af+Bg=\mathrm{Res}(f,g,x_1)$;
and (ii) $f,g$ have a common factor with positive degree in $x_1$ if and only if this resultant vanishes (see e.g. Proposition 1, section \S 6 of chapter 3 in \cite{clo}). This extends directly to the case at hand, where the ring of coefficients $\cal R$ is a unique factorisation domain, by working in the corresponding field of fractions. 

 As our inductive hypothesis, we assume that $(s_k,t_k)=1$ for $0\leq k\leq n$. The base cases $k=0,1$ are trivial, while for $k=2$ we can write 
\beq \label{linsys}
\left(\begin{array}{c} r_0s_2 \\ r_0 t_2 \end{array}\right) 
=\left(\begin{array}{cc} -F^{(2)}(s_1,t_1) & F^{(1)}(s_1,t_1) \\ 
-F^{(3)}(s_1,t_1) & F^{(2)}(s_1,t_1) \end{array}\right) 
\left(\begin{array}{c} s_0 \\ t_0 \end{array}\right) , 
\eeq 
and verify directly that  
$ \rN_2=r_0s_2$ and  
$\hat{\rN}_2=r_0t_2$ are coprime polynomials in ${\cal R} [s_0,t_0,s_1,t_1]$, and homogeneous in $s_0,t_0$ and $s_1,t_1$ separately (of degree 1 and 4, respectively). Moreover, if we regard $ \rN_2,\hat{\rN}_2$ as (linear) polynomials in 
$u_0=s_0/t_0$,  then from the aforementioned facts their coprimality means that there are polynomials $A,B$ (also linear in $u_0$), whose 
coefficients are homogeneous polynomials in 
${\cal R}[s_1,t_1]$, such that 
$$
t_0^{-1}(Ar_0s_2+Br_0t_2)=\mathrm{Res}(t_0^{-1}r_0s_2,t_0^{-1}r_0t_2,u_0)\neq 0.
$$
As it happens, from the linearity of the system (\ref{linsys}) we see that the resultant in this case is just the determinant of the $2\times 2$ matrix on the right, and turns out to be equal to the polynomial $-Q(s_1,t_1)$,   homogeneous of degree 8, so after rescaling by $t_0^2$ we have 
$$
A'r_0s_2+B'r_0t_2=- t_0^2 Q(s_1,t_1), 
$$ 
where $A',B'$ are homogeneous in  $s_0,t_0$ and $s_1,t_1$ (separately). Upon shifting indices, by 
applying the pullback of the map 
defined by the Laurent system (\ref{KP1}), this gives an identity of Laurent polynomials for all $n$, namely 
\beq\label{laurid} 
A'(s_{n-1},t_{n-1},s_n,t_n) r_{n-1} s_{n+1}+
B'(s_{n-1},t_{n-1},s_n,t_n) r_{n-1} t_{n+1}= -t_{n-1}^2  Q(s_n,t_n).
\eeq 

Now we suppose that $s_{n+1}$ and $t_{n+1}$ have a non-trivial common factor $P$, and show that 
under the inductive hypothesis this leads to a contradiction. Indeed, from the right-hand side of (\ref{laurid}) there are two possibilities: (a) $P| t_{n-1}$, or (b) $P| Q(s_n,t_n)$. 
In case (a), (\ref{KP1})  directly yields 
$$ 
r_{n-1}s_{n+1} -t_{n-1}F^{(1)}(s_n,t_n) =  -s_{n-1} F^{(2)}(s_n,t_n), 
\, \,
r_{n-1}t_{n+1} -t_{n-1}F^{(2)}(s_n,t_n) =  -s_{n-1} F^{(3)}(s_n,t_n), 
$$ 
therefore $P$ must divide both right-hand sides above, 
and, since $(s_{n-1},t_{n-1})=1$ by the inductive hypothesis, this yields 
$P| F^{(2)}(s_n,t_n)$ and $F^{(3)}(s_n,t_n)$. By applying the same argument as in the proof of Proposition \ref{thm5}, we form the Sylvester matrix corresponding to the coefficients of $F^{(2)}$ and $F^{(3)}$, whose determinant is a non-zero element of $\cal R$, and since also 
$(s_{n},t_{n})=1$ by  hypothesis, this gives a contradiction. Thus  
we are left with case (b). The fact that $s_2,t_2$ are both coprime 
to $Q(s_1,t_1)$ can be verified directly, so by taking resultants 
with respect to $u_1$ and shifting indices this leads to a pair of identities 
$$ 
A^*_n r_{n-1} s_{n+1}+B^*_n Q(s_n,t_n)=t_n^MR^*(s_{n-1},t_{n-1}), 
\quad 
 A^\dagger_n r_{n-1} s_{n+1}+B^\dagger_n Q(s_n,t_n)=t_n^MR^\dagger(s_{n-1},t_{n-1}), 
$$ 
where $A_n^*,B_n^*,A_n^\dagger, B^\dagger_n$ are 
homogeneous polynomials in $s_{n-1},t_{n-1},s_n,t_n$, $M$ is a positive integer, and the resultants $R^*,R^\dagger$ are a pair of coprime homogeneous polynomials in their arguments. Using the same method as for Proposition \ref{thm5} once again, the latter coprimality means that, 
since  $(s_{n-1},t_{n-1})=1$ by hypothesis, $P$ cannot divide both 
$R^*(s_{n-1},t_{n-1})$ and $R^\dagger(s_{n-1},t_{n-1})$, hence 
$P|t_n$. To finish off the argument, it is enough to use 
$(s_2,t_1)=1=(t_2,t_1)$ and take resultants with respect to $v_1=t_1/s_1$, then shift to obtain further relations of 
the form 
$$ 
\tilde{A}_n r_{n-1} s_{n+1}+\tilde{B}_n t_n=s_n^L\tilde{R}(s_{n-1},t_{n-1}), 
\quad 
 \hat{A}_n r_{n-1} s_{n+1}+\hat{B}_n t_n=s_n^L\hat{R}(s_{n-1},t_{n-1}) 
$$ 
for all $n$, where $L$ is a positive integer and the resultants $\tilde{R},\hat{R}$ are coprime homogeneous polynomials in their arguments. 
From $(s_n,t_n)=1$ it follows from the two right-hand sides above that $P|\tilde{R}(s_{n-1},t_{n-1}), 
\hat{R}(s_{n-1},t_{n-1})$, which is seen to be a contradiction by applying the argument used for Proposition \ref{thm5} once more. This 
completes the proof of Lemma \ref{coprime}. 
\hfill $\qed$

Analogous arguments can be used to obtain similar coprimality conditions for all Laurent polynomials  
$r_n,s_n,t_n$ generated by (\ref{KP1}), and it appears that more is true: these iterates should be distinct  irreducible elements of $\hat{\cal R}={\cal R}[r_0^{\pm 1}, s_0, s_1, t_0, t_1]$ for all $n$.

\end{document}